\documentclass[a4paper]{article}

\usepackage{fullpage}
\usepackage{hyperref}
\usepackage[title]{appendix}

\usepackage{amsmath}
\usepackage{amssymb}
\usepackage{bbm}
\usepackage{enumitem}

\usepackage{algorithm}
\usepackage{algpseudocode}
\usepackage{varwidth} 

\usepackage{amsthm}
\newtheorem{theorem}{Theorem}
\newtheorem{lemma}{Lemma}
\newtheorem{fact}{Fact}

\theoremstyle{definition}
\newtheorem{definition}{Definition}

\usepackage{algorithm}
\usepackage{algpseudocode}
\usepackage{varwidth} 

\usepackage{graphicx}
\graphicspath{ {./} }

\usepackage{marginnote}

\newcommand{\expected}[1]{\mathop{{}\mathbb{E}}\left[{#1}\right]}
\newcommand{\proba}[1]{\Pr\left[\,{#1}\,\right]}

\DeclareMathOperator*{\argmin}{arg\,min}

\newif\ifappendix
\appendixtrue

\begin{document}

\title{Limits of Sequential Local Algorithms on the Random $k$-XORSAT Problem}
\author{Kingsley Yung \thanks{The Chinese University of Hong Kong. Email: \texttt{koyung@cse.cuhk.edu.hk}.}}
\date{}
\maketitle

\begin{abstract}
    The random $k$-XORSAT problem is a random constraint satisfaction problem of $n$ Boolean variables and $m=rn$ clauses, which a random instance can be expressed as a $G\mathbb{F}(2)$ linear system of the form $Ax=b$, where $A$ is a random $m \times n$ matrix with $k$ ones per row, and $b$ is a random vector.
    It is known that there exist two distinct thresholds $r_{core}(k) < r_{sat}(k)$ such that as $n \rightarrow \infty$ for $r < r_{sat}(k)$ the random instance has solutions with high probability, while for $r_{core} < r < r_{sat}(k)$ the solution space shatters into an exponential number of clusters.
    Sequential local algorithms are a natural class of algorithms which assign values to variables one by one iteratively.
    In each iteration, the algorithm runs some heuristics, called local rules, to decide the value assigned, based on the local neighborhood of the selected variables under the factor graph representation of the instance.

    We prove that for any $r > r_{core}(k)$ the sequential local algorithms with certain local rules fail to solve the random $k$-XORSAT with high probability.
    They include (1) the algorithm using the Unit Clause Propagation as local rule for $k \ge 9$, and (2) the algorithms using any local rule that can calculate the exact marginal probabilities of variables in instances with factor graphs that are trees, for $k\ge 13$.
    The well-known Belief Propagation and Survey Propagation are included in (2).
    Meanwhile, the best known linear-time algorithm succeeds with high probability for $r < r_{core}(k)$.
    Our results support the intuition that $r_{core}(k)$ is the sharp threshold for the existence of a linear-time algorithm for random $k$-XORSAT.

    Our approach is to apply the Overlap Gap Property OGP framework to the sub-instance induced by the core of the instance, instead of the whole instance.
    By doing so, the sequential local algorithms can be ruled out at density as low as $r_{core}(k)$, since the sub-instance exhibits OGP at much lower clause density, compared with the whole instance.

\end{abstract}

\section{Introduction}

\subsection{Background}

The $k$-XORSAT problem is a Boolean constraint satisfaction problem and closely related to the well-known $k$-SAT problem.
An instance $\Phi$ of the $k$-XORSAT problem consists of $m$ clauses in $n$ Boolean variables.
Each clause is a Boolean linear equation of $k$ variables of the form $x_{j_1} \oplus x_{j_2} \oplus \cdots \oplus x_{j_k} = b_j$, where $\oplus$ is the modulo-2 addition.
By convention, when we say an XORSAT instance, without the prefix "$k$", we mean the same except we do not require the clauses to have exactly $k$ variables.
An \textbf{assignment} $\sigma$ to the $n$ variables is a mapping from the set $\{x_i: i \in [n]\}$ of all $n$ variables to the set $\{0,1\}$ of the two Boolean values.
By abusing the notation, we can write it as the Boolean vector $\sigma = (\sigma(x_1), \sigma(x_2), \cdots, \sigma(x_n)) \in \{0,1\}^n$ containing the assigned values.
The \textbf{distance} $d(\sigma,\sigma')$ between any two assignments $\sigma$ and $\sigma'$ is defined to be the Hamming distance $d(\sigma,\sigma') = \sum_{i=1}^{n} \mathbbm{1}(\sigma(x_i) \neq \sigma'(x_i))$.
A clause is \textbf{satisfied by} an assignment if the equation of the clause holds when the variables are replaced by the corresponding assigned values, and an instance of the $k$-XORSAT problem is \textbf{satisfied by} an assignment if all its clauses are satisfied by the assignment.
An instance is \textbf{satisfiable} if it has at least one satisfying assignment, or \textbf{unsatisfiable} if it does not have any satisfying assignment.
The assignment $\sigma$ that satisfies the instance $\Phi$ is called a \textbf{solution} for the instance.
The set of all satisfying solutions for the instance $\Phi$ is called the \textbf{solution space} of the instance, denoted by $\mathcal{S}(\Phi)$.
We are interested in the complexity of finding a solution.

Since each clause is just a Boolean linear equation, an instance $\Phi$ can be viewed as a Boolean linear system $Ax=b$, where $A \in \{0,1\}^{m \times n}$ is an $m \times n$ Boolean matrix and $b \in \{0,1\}^n$ is a vector of length $n$.
Note that each row in $A$ contains exactly $k$ ones, since each clause has exactly $k$ variables.
We can see that finding solutions for a $k$-XORSAT instance is equivalent to solving a Boolean linear system, and the solution space $\mathcal{S}(\Phi)$ is an affine space inside $\{0,1\}^n$.
By abusing the notation, we can simply write $\Phi=(A,b)$, and the terms "clause" and "equation" are interchangeable here.

We are particularly interested in random instances of the $k$-XORSAT problem.
In a random instance $\mathbf{\Phi}$, each clause is drawn over all $2\binom{n}{k}$ possibilities, independently.
In particular, the left-hand side of the equation is the modulo-2 sum of $k$ variables chosen uniformly from $\binom{n}{k}$ possibilities, and the right-hand side is either 0 or 1 with even probabilities.
Therefore, a random instance $\mathbf{\Phi}$ of the $k$-XORSAT problem is drawn uniformly from the ensemble $\mathbf{\Phi}_k(n,m)$ of all possible instances of the $k$-XORSAT problem with $n$ variables and $m$ clauses, each clause containing exactly $k$ variables, and we denote this by $\mathbf{\Phi} \sim \mathbf{\Phi}_k(n,m)$.
We focus on the regime in which the number of variables $n$ goes to the infinity and the number of clauses $m$ is proportional to the number of variables $n$, that is, $m = rn$, where $r$ is a constant independent of $n$ and called the \emph{clause density}.

Since a $k$-XORSAT instance can be represented by a system of linear equations in $G\mathbb{F}(2)$, given an instance, some standard linear algebra methods such as the Gaussian elimination can determine whether the instance is satisfiable, find a solution, and even count the total number of solutions, in polynomial time.
However, beyond this particular algebraic structure, some variants of the $k$-XORSAT problem is hard to solve.
Achlioptas and Molloy mentioned in their paper \cite{achlioptasSolutionSpaceGeometry2015} that random instances of the $k$-XORSAT problem seems to be extremely difficult for both generic CSP solvers and SAT solvers, which do not take the algebraic structure into account.
Guidetti and Young \cite{guidettiComplexitySeveralConstraintsatisfaction2011} suggested that the random $k$-XORSAT is the most difficult for random walk type algorithms such as WalkSAT, among many random CSPs.
The difficulty of solving random $k$-XORSAT instances becomes more apparent when we only consider linear-time algorithms as efficient algorithms, since we do not have linear-time algebraic method to solve a linear system in general.

Many studies suggest that the difficulties of solving random CSPs are related to the \emph{phase transition} of the solution spaces when the clause density $r$ grows.
(We will have more detailed discussion in Section \ref{sec:phase-transition-of-random-k-xorsat}.)
Pittel and Sorkin \cite{pittelSatisfiabilityThresholdXORSAT2016} obtained the sharp \emph{satisfiability threshold} $r_{sat}(k)$ of random $k$-XORSAT, for general $k \ge 3$:
The random $k$-XORSAT instance is satisfiable w.h.p. when $r < r_{sat}(k)$, and it is unsatisfiable w.h.p. when $r > r_{sat}(k)$.
(We say an event $\mathcal{E}_n$, depending on a number $n$, occurs \textbf{with high probability}, or shortened to \textbf{w.h.p.}, if the probability of the event $\mathcal{E}_n$ occurring converges to 1 as $n$ goes to the infinity, that  is, $\lim_{n \rightarrow \infty}\proba{\mathcal{E}_n} = 1$.)
Furthermore, Ibrahimi, Kanoria, Kraning and Montanari \cite{ibrahimiSetSolutionsRandom2012} obtained the sharp \emph{clustering threshold} $r_{core}(k)$, which is less than $r_{sat}(k)$, of random $k$-XORSAT for $k \ge 3$.
When $r < r_{core}(k)$, w.h.p. the solution space of a random $k$-XORSAT instance is "well-connected".
When $r_{core}(k) < r < r_{sat}(k)$, w.h.p. the solution space of a random $k$-XORSAT instance shatters into an exponential number of "well-separated clusters".
In \cite{ibrahimiSetSolutionsRandom2012}, They also provided a linear-time algorithm that can solve a random $k$-XORSAT instance w.h.p. for $r < r_{core}(k)$.
On the other hand, no algorithm is known to be able to find a solution for a random $k$-XORSAT instance with non-vanishing probability in linear time, for $r_{core}(k) < r < r_{sat}(k)$ in which solutions exist with high probability.

In this work, we consider a natural class of algorithms, called \emph{sequential local algorithms}.
A sequential algorithm selects an unassigned variable randomly and assigns a value to it, iteratively, until every variable has an assigned value.
In each iteration of the algorithm, to decide the assigned value, the algorithm runs some heuristic called \emph{local rules} which return a value $p\in[0,1]$, and decide the assigned value to be either 0 or 1 randomly, according to the Bernoulli distribution with parameter $p$.
Ideally, if in each iteration the local rule can calculate the exact marginal probability of the selected variable over a randomly chosen solution for the instance conditioned on fixing all previously selected variables to their assigned values, the algorithm should be able to find a solution.
However, we restrict the ability of the local rules by only providing \emph{local} structure to the local rules.
To explain the meaning of "local", we first construct a graphical representation for the $k$-XORSAT instances: the \emph{factor graph}.
The \emph{factor graph} $G$ of a $k$-XORSAT instance $\Phi$ is constructed in the following way: (1) each variable is represented by a \emph{variable node}; (2) each equation is represented by an \emph{equation node}; (3) add an undirected edge $(v,e)$ if the variable $v$ is involved in the equation $e$.
Note that since there is a one-to-one correspondence between variables (equations) and variable nodes (equation nodes), in this paper, the terms \emph{variables (equations)} and \emph{variable nodes (equation nodes)} are interchangeable.
The \emph{distance} between any two nodes is the number of edges in the shortest path connecting the two nodes.
For any integer $R \ge 0$, the \emph{local neighborhood} $B_G(v,R)$ with radius $R$ of a node $v$ is the subgraph of $G$ induced by all nodes with distances less than or equal to $R$ from the node $v$.
By "local" in the name of the algorithms, it means the local rules only takes the local neighborhood of the selected variable, of radius $R$, as its input.

The actual implementation of a sequential local algorithm depends on the choice for the local rules.
To emphasize the choices for the local rules of the algorithms, the sequential local algorithm with the given local rule $\tau$ is called the \emph{$\tau$-decimation algorithm} \texttt{DEC\textsubscript{$\tau$}}.
The formal definitions of the sequential local algorithms, as well as the $\tau$-decimation algorithms, will be given in Section \ref{sec:sequential-local-algorithms}.
Note that if the local rule $\tau$ takes constant time, then the $\tau$-decimation algorithm also takes linear time.

We introduce a notion of \emph{freeness} to the sequential local algorithms.
For any iteration in which the local rule returns $1/2$, we call it a \emph{free step}.
Intuitively, a free step means the local rule cannot obtain useful information from the local structure and let the algorithms make a random guess for the assigned value.
We say a $\tau$-decimation algorithm is \emph{$\delta$-free} if w.h.p. it has at least $\delta n$ free steps.
Moreover, we say a $\tau$-decimation algorithm is \emph{strictly $\delta$-free} if it is $\delta'$-free for some $\delta'>\delta$.
If the $\tau$-decimation algorithm is $\delta$-free with large $\delta>0$, it means the algorithm makes a lot of random guess on the assigned values, and it is likely that the local rule $\tau$ cannot extract useful information from the local structure to guide the algorithm.
This leads to our contribution described in the next section.

\subsection{Main contribution}

The main contribution of this work consists of two parts.
The first part is to show that as $n \rightarrow \infty$ if the $\tau$-decimation algorithm is strictly $2\mu(k,r)$-free then w.h.p. it fails to find a solution for the random $k$-XORSAT instance, when the clause density $r$ is beyond the clustering threshold $r_{core}(k)$ but below the satisfiability threshold $r_{sat}(k)$.
This can be formally written as Theorem \ref{thm:main-theorem}.
The value $\mu(k,r)$ given in Theorem \ref{thm:main-theorem} is an upper bound of the number of variables removed from the instance in order to obtain the sub-instance called \emph{core instance}, which is crucial in our proof.
We will discuss its meaning in detail in Section \ref{sec:overlap-gap-property}.

\begin{theorem}
    \label{thm:main-theorem}
    For any $k \ge 3$ and $r \in (r_{core}(k),r_{sat}(k))$, if the $\tau$-decimation algorithm \textnormal{\texttt{DEC\textsubscript{$\tau$}}} is strictly $2\mu(k,r)$-free on the random $k$-XORSAT instance $\mathbf{\Phi} \sim \mathbf{\Phi}_k(n,rn)$,
    then w.h.p. the output assignment $\textnormal{\texttt{DEC\textsubscript{$\tau$}}}(\mathbf{\Phi})$ from the algorithm $\textnormal{\texttt{DEC\textsubscript{$\tau$}}}$ on input $\mathbf{\Phi}$ is not in the solution space $\mathcal{S}(\mathbf{\Phi})$, that is,
    \begin{align*}
        \lim_{n \rightarrow \infty} \proba{\textnormal{\texttt{DEC\textsubscript{$\tau$}}}(\mathbf{\Phi}) \in \mathcal{S}(\mathbf{\Phi})} = 0,
    \end{align*}
    where
    $Q$ is the largest solution of the fixed point equation $Q = 1 - \exp(-krQ^{k-1})$, and
    $\mu(k, r)$ is the real-valued function given by
        $\mu(k, r) = \exp(-krQ^{k-1}) + krQ^{k-1} \exp(-krQ^{k-1})$.
\end{theorem}

\noindent \textbf{Note.\;} This theorem can also be applied to the \emph{non-sequential local algorithms} in which the algorithms run their local rules and decide the assigned value for each variable, \emph{in parallel}, without depending on the other assigned values.
We will briefly discuss the reason at the end of Section \ref{sec:technique} where we discuss the proof technique we used.

\medskip

The best known linear algorithm of finding a solution for the random $k$-XORSAT instance succeeds w.h.p., for $k \ge 3$ and $r < r_{core}(k)$ \cite{ibrahimiSetSolutionsRandom2012}.
That means these sequential local algorithms do not outperform the best known linear algorithm.
Note that $r_{core}(k)$ is where the best known linear algorithm succeeds up to, and where the sequential local algorithms starts failing.
These support the intuition that $r_{core}(k)$ is the sharp threshold of the existence of a linear-time algorithm for random $k$-XORSAT problem.

The second part of our contribution is to verify that the "freeness" condition in Theorem \ref{thm:main-theorem} is satisfied by the $\tau$-decimation algorithm with certain local rules $\tau$.
One of them is the simplest local rule, the Unit Clause Propagation \texttt{UC}, which tries to satisfy the unit clause on the selected variable if exists, or make random guess otherwise.
By using the Wormald's method of differential equations to count the number of free steps run by \texttt{UC}-decimation algorithm \texttt{DEC\textsubscript{UC}}, we can show that it is strictly $2\mu(k,r)$-free on the random $k$-XORSAT instance $\mathbf{\Phi}$ for $k \ge 9$, which leads to the following theorem.
\begin{theorem}
    \label{thm:main-theorem-unit-clause}
    For any $k \ge 9$, $r \in (r_{core}(k),r_{sat}(k))$,
    given a random $k$-XORSAT instance $\mathbf{\Phi} \sim \mathbf{\Phi}_k(n,rn)$,
    we denote by \textnormal{\texttt{DEC\textsubscript{UC}}}$(\mathbf{\Phi})$ the output assignment from the \textnormal{\texttt{UC}}-decimation algorithm \textnormal{\texttt{DEC\textsubscript{UC}}} on input $\mathbf{\Phi}$.
    Then, we have
    $\lim_{n \rightarrow \infty} \proba{\textnormal{\texttt{DEC\textsubscript{UC}}}(\mathbf{\Phi}) \in \mathcal{S}(\mathbf{\Phi})} = 0.$
\end{theorem}

In each iteration, the role of the local rules is to calculate the marginal probability of the selected variable in the instance conditioned on fixing all previously selected variables to their assigned values.
Belief Propagation \texttt{BP} and Survey Propagation \texttt{SP} are surprisingly good at approximating marginal probabilities of variables over randomly chosen solutions in many random constraint satisfaction problems empirically \cite{mezardAnalyticAlgorithmicSolution2002, gomesSatisfiedPhysics2002, braunsteinSurveyPropagationAlgorithm2005, ricci-tersenghiCavityMethodDecimated2009, krocMessagepassingLocalHeuristics2009}.
In particular, it is well-known that they can calculate the exact marginal probabilities of variables when the underlying factor graph is a tree, which is proved analytically.
If Belief Propagation \texttt{BP} and Survey Propagation \texttt{SP} are used as the local rule $\tau$, it is natural to expect that the $\tau$-decimation algorithm can find a solution.
However, we prove that even the local rule $\tau$ can give the exact marginal probabilities of variables over a randomly chosen solution for any instance whose factor graph is a tree, the $\tau$-decimation algorithm still cannot find a solution w.h.p. for $k \ge 13$.
We know that w.h.p. the local neighborhood of the factor graph of the random $k$-XORSAT instance is a tree.
Therefore, running \texttt{BP} and \texttt{SP} on the local neighborhood actually gives the exact marginal probabilities of the selected variables, with respect to the sub-instance induced by the local neighborhood.
This implies that both \texttt{BP}-decimation algorithm \texttt{DEC\textsubscript{BP}} and \texttt{SP}-decimation algorithm \texttt{DEC\textsubscript{SP}} fail to find a solution w.h.p. for $k \ge 13$.
\begin{theorem}
    \label{thm:main-theorem-exact-marginal}
    For any $k \ge 13$, $r \in (r_{core}(k),r_{sat}(k))$,
    given a random $k$-XORSAT instance $\mathbf{\Phi} \sim \mathbf{\Phi}_k(n,rn)$,
    denote by \textnormal{\texttt{DEC\textsubscript{$\tau$}}}$(\mathbf{\Phi})$ the output assignment from the \textnormal{\texttt{$\tau$}}-decimation algorithm \textnormal{\texttt{DEC\textsubscript{UC}}} on input $\mathbf{\Phi}$.
    Assume the local rule $\tau$ outputs the exact marginal probability of a selected variable for any instance whose factor graph is a tree.
    Then, we have
    $\lim_{n \rightarrow \infty} \proba{\textnormal{\texttt{DEC\textsubscript{$\tau$}}}(\mathbf{\Phi}) \in \mathcal{S}(\mathbf{\Phi})} = 0.$
\end{theorem}

To prove Theorem \ref{thm:main-theorem-unit-clause} and Theorem \ref{thm:main-theorem-exact-marginal}, we only need to calculate the number of free steps in \texttt{DEC\textsubscript{UC}} and the number of free steps in \texttt{DEC\textsubscript{$\tau$}} with the assumption on $\tau$ described in Theorem \ref{thm:main-theorem-exact-marginal}.
If the number of free steps is strictly greater than $2\mu(k,r)n$ with high probability, we know that they are strictly $2\mu(k,r)$-free, and the results follow immediately by applying Theorem \ref{thm:main-theorem}.
Similarly, to obtain the same results for other $\tau$-decimation algorithms, all we need to do is to calculate the number of free steps for those algorithms.
If they are strictly $2\mu(k,r)$-free, we can obtain the same results by applying Theorem \ref{thm:main-theorem}.
Note that, due to certain limitations of our calculation, our results are limited to $k \ge 9$ in Theorem \ref{thm:main-theorem-unit-clause} and $k \ge 13$ in Theorem \ref{thm:main-theorem-exact-marginal}.
We believe that the results hold for general $k \ge 3$, and can be proved by improving some subtle calculation in our argument.

Although we only show a few of implementations of the sequential local algorithms, we believe that the results are general across many sequential local algorithms with different local rules due to Theorem \ref{thm:main-theorem-exact-marginal}.
In the framework of sequential local algorithms, the role of the local rules is to approximate the marginal probabilities of the selected variables over a random solution for the instance induced by the local neighborhood centered at the selected variables.
Therefore, we believe that for any local rule that can make a good approximation on the marginals, it shall give similar results as Theorem \ref{thm:main-theorem-exact-marginal}.
(Note that a more general definition of "$\delta$-free" may be useful, for example, we can say a $\tau$-decimation algorithm is \emph{$(\delta,\epsilon)$-free} if we have $|p-1/2|<\epsilon$, where $p$ is the value returned by the local rule, for at least $\delta n$ iterations.)

It is worth to mention the differences between these implementations of the sequential local algorithms and their well-known variants.
Firstly, the \texttt{UC}-decimation algorithm is slightly different from the well-known Unit Clause algorithm.
Under the framework of sequential local algorithm, the variables are selected in a random order.
However, in the well-studied Unit Clause algorithm the variables in unit clauses are selected first \cite{achlioptasTwocoloringRandomHypergraphs2002}.
The difference in the variable order could be crucial to the effectiveness of the Unit Clause algorithm.
Secondly, the \texttt{BP}-decimation algorithm and the \texttt{SP}-decimation algorithm are slightly different from the Belief Propagation Guided Decimation Algorithm \cite{coja-oghlanBeliefPropagationGuided2017} and Survey Propagation Decimation Algorithm \cite{braunsteinSurveyPropagationLocal2004, braunsteinSurveyPropagationAlgorithm2005, manevaNewLookSurvey2007}.
In the framework of sequential local algorithms, we only provide the local neighborhood to \texttt{BP} and \texttt{SP}.
It is equivalent to bounding the number of messages passed in each decimation step by the constant $R \ge 0$ in BP-guided Decimation Algorithm and SP-guided Decimation Algorithm.
It is in contrast to many empirical studies of BP-guided Decimation Algorithm and SP-guided Decimation Algorithm which allow the message passing iteration to continue until it converges.


Moreover, this work provides a new variant of the \emph{overlap gap property method}, which was originally introduced by Gamarnik and Sudan \cite{gamarnikPerformanceSequentialLocal2017}.
Instead of considering the overlap gap property of the whole instance, we utilize that property of a sub-instance of the random $k$-XORSAT instance.
In particular, the proof of this work is inspired by \cite{gamarnikPerformanceSequentialLocal2017}, which uses the overlap gap property method to rule out the class of balanced sequential local algorithms from being able to solve random NAE-$k$-SAT problem when the clause density is close to the satisfiability threshold.
Instead of directly applying the method on the whole instance, we focus on the sub-instance induced by the 2-core of the factor graph of the instance.
This modification help us obtain tight bounds of algorithmic threshold unlike \cite{gamarnikPerformanceSequentialLocal2017}.
If we apply the original overlap gap property method and use the first moment method to obtain the property, we are able to show that the sequential local algorithms fail to solve the random $k$-XORSAT problem when the clause density is lower than a certain threshold $r_1(k)$.
However, that threshold $r_1(k)$ is much higher than the $r_{core}(k)$.
It only tells us that the algorithms fails when the density is very close to the satisfiability threshold $r_{sat}(k)$.
With our modification, we can lower that threshold to exactly $r_{core}(k)$, namely, the algorithms fail in finding a solution when the clause density is as low as the clustering threshold.
This opens a new possibility to improve other results which use the overlap gap property method on other random constraint satisfaction problems.

\subsection{Phase transition of random $k$-XORSAT}
\label{sec:phase-transition-of-random-k-xorsat}

Many random ensembles of constraint satisfaction problems CSPs such as random $k$-SAT and random NAE-$k$-SAT are closely related to the random $k$-XORSAT.
For example, the well-known random $k$-SAT is analogous to the random $k$-XORSAT, in the sense that we can obtain a $k$-XORSAT instance from a $k$-SAT instance by replacing OR operators with XOR operators.
We are particularly interested in the existences of some sharp thresholds on the clause density $r$ in which the behaviors of a random instance changes sharply when the clause density $r$ grows and passes through those thresholds.
The following three thresholds are particularly of interest.
\begin{enumerate}
    \item The \emph{satisfiability threshold} separates the regime where w.h.p. the random instance is satisfiable and the regime where w.h.p. it is unsatisfiable.
    \item The \emph{clustering threshold} separates the regime where w.h.p. the solution space can be partitioned into well-separated subsets, each containing an exponential small fraction of solutions, and the regime where w.h.p. the solution space cannot be partitioned in this way.
    \item The \emph{algorithmic threshold} separates the regime where we have an efficient algorithm that can find a solution for a satisfiable random instance with non-vanishing probability, and the regime where no such algorithm exists.
\end{enumerate}
Many random constraint satisfaction problems such as random $k$-SAT, random NAE-$k$-SAT and random graph coloring share the following phenomena related to these thresholds \cite{achlioptasAlgorithmicBarriersPhase2008}.
\begin{itemize}
    \item There is an upper bound of the (conjectured) satisfiability threshold.
    \item There is a lower bound of the (conjectured) satisfiability threshold, from the non-constructive proof, and the lower bound is essentially tight.
    \item There are some polynomial time algorithms that can find a solution when the density is relatively low, but no polynomial time algorithm is known to succeed when the density is close to the satisfiability threshold. This leads to a conjectured algorithmic threshold, which is asymptotically below the (conjectured) satisfiability threshold.
    \item The clustering phenomenon takes place when the density is greater than a (conjectured) clustering threshold, and this threshold is close to or even asymptotically equal to the algorithmic threshold.
\end{itemize}

It is worth to mention that not every random constraint satisfaction problems share this set of phenomena.
The most notable example is \emph{symmetric binary perceptron} (SBP).
Its satisfiability threshold $\alpha_{sat}(k)>0$ was established by \cite{perkinsFrozen1RSBStructure2021, abbeProofContiguityConjecture2022}.
They also showed that SBP exhibits clustering property and almost all clusters are singletons, for clause density $\alpha > 0$.
On the other hand, \cite{abbeBinaryPerceptronEfficient2022} gave a polynomial-time algorithm that can find solutions w.h.p., for low clause density.
Therefore, there is a regime of low clause density in which SBP exhibits clustering property, and it is solvable in polynomial time, simultaneously.
Its clustering phenomenon does not cause the hardness.

Being analogous to the random $k$-SAT problem, the random $k$-XORSAT problem shares those phenomena with other random constraint satisfaction problems.
However, the story is slightly different in the random $k$-XORSAT problem.
Since the $k$-XORSAT instances are equivalent to Boolean linear systems, their solution spaces are simply some affine subspaces in the Hamming hypercube $\{0,1\}^n$.
Because of their algebraic structures, we are able to obtain the existence and the ($n$-independent) value of the satisfiability threshold of the random $k$-XORSAT problem.
Dubois and Mandler \cite{dubois3XORSATThreshold2002} proved that there exists an $n$-independent satisfiability threshold $r_{sat}(k)$ for $k=3$, and determined the exact value of the sharp threshold by the second moment argument.
Pittel and Sorkin \cite{pittelSatisfiabilityThresholdXORSAT2016} further extended it for general $k \ge 3$.
An independent work on cuckoo hashing from \cite{dietzfelbingerTightThresholdsCuckoo2010} also included an argument of the $k$-XORSAT satisfiability threshold for general $k \ge 3$.
Those proofs consider the 2-core of the hypergraph associated with the random $k$-XORSAT instance.
(In graph theory, a \emph{$k$-core} of a (hyper)graph is a maximal subgraph in which all vertices have degree at least $k$.)
Based on the 2-core, we can construct a sub-instance, called \emph{2-core instance} or simply \emph{core instance} of the original instance.
One can prove that the original instance is satisfiable if and only if core instance is satisfiable.
\cite{coccoRigorousDecimationBasedConstruction2003, mezardTwoSolutionsDiluted2003} studied the core instance and determined the satisfiability threshold of the core instance, which can be converted to the satisfiability threshold of the random $k$-XORSAT instance.

Mézard, Ricci-Tersenghi and Zecchina \cite{mezardTwoSolutionsDiluted2003} started the study of the clustering phenomenon of random $k$-XORSAT, and linked it to the existence of the non-empty 2-core instance.
From \cite{pittelSuddenEmergenceGiant1996, molloyCoresRandomHypergraphs2005, kimPoissonCloningModel2004}, we know the non-empty 2-core of random hypergraphs suddenly emerges at a critical edge density $r_{core}(k)$.
After that, Ibrahimi, Kanoria, Kraning and Montanari \cite{ibrahimiSetSolutionsRandom2012}, and Achlioptas and Molloy \cite{achlioptasSolutionSpaceGeometry2015} independently proved that there exists the clustering threshold $r_{clt}(k)$, which is equal to $r_{core}(k)$ and smaller than $r_{sat}(k)$, such that w.h.p. the solution space is a connected component for density $r < r_{core}(k)$, and w.h.p. the solution space shatters into exponentially many $\Theta(n)$-separated components for density $r > r_{core}(k)$, provided that we consider the solution space as a graph in which we add an edge between two solutions if their Hamming distance is $O(\log n)$.

As we mentioned before, the random $k$-XORSAT instance can be written as a random Boolean linear system, so it can be solved in polynomial time by using linear algebra method, regardless the clause density.
For example, Gaussian elimination can solve it in $O(n^3)$ time.
Since we do not have linear time algebraic method to solve linear system, we can still study the algorithmic phase transition if we only consider linear-time algorithms as efficient algorithms.
In the proofs in \cite{ibrahimiSetSolutionsRandom2012}, they provided an algorithm that can find a solution in linear time when $r < r_{core}(k)$, which implies that $r_{core}(k)$ is a lower bound of the (linear-time version of) algorithmic threshold $r_{alg}(k)$ of the random $k$-XORSAT problem.
We conjecture that no algorithm can solve the random $k$-XORSAT problem in linear time with non-vanishing probability when $r > r_{core}(k)$, which implies $r_{core}(k)$ is an upper bound of $r_{alg}(k)$ and thus $r_{alg}(k)=r_{core}(k)$.
This would lead to the intimate relation between the failure of linear time algorithms on random $k$-XORSAT and the clustering phenomenon of its solution space.

\subsection{Sequential local algorithms}
\label{sec:sequential-local-algorithms}

Sequential local algorithms are a class of algorithms parametrized by a local rule $\tau$ that specifies how values should be assigned to variables based on the "neighborhoods" of the variables.
Given a local rule $\tau$, the sequential local algorithm can be written as the following \emph{$\tau$-decimation algorithm}.

Given a fixed even number $R \ge 0$, we denote by $\mathcal{I}_R$ the set of all instances in which each of those instances has exactly one of its variables selected as \emph{root}, and all nodes in its factor graph have distances from the root variable node at most $R$.
A \textbf{local rule} is defined to be a function $\tau:\mathcal{I}_R \rightarrow [0,1] \in \mathbb{R}$, mapping from $\mathcal{I}_R$ to the interval $[0,1]$.
Given an instance $\Phi$, since the local neighborhood $B_\Phi(x^*,R)$ of a variable node $x^*$ represents a sub-instance of $\Phi$ induced by all nodes having distance at most $R$ from the root variable node $x^*$, we have $B_\Phi(x^*,R) \in \mathcal{I}_R$ and $\tau(B_\Phi(x^*,R))$ is well-defined.
Then, the $\tau$-decimation algorithm can be expressed as the followings.

\begin{algorithm}[H]
    \caption{$\tau$-decimation algorithm}
    \begin{algorithmic}[1]
        \State Input: \begin{varwidth}[t]{\linewidth}
            an instance of the k-XORSAT problem $\Phi$, \\
            an even number $R \ge 0$, and \\
            a local rule $\tau:\mathcal{I}_R \rightarrow [0,1]$.
        \end{varwidth}
        \State Set $\Phi_0 = \Phi$.
        \For {$t=0,...,n-1$}
            \State Select an unassigned variable $x^*$ from $\Phi_{t}$, uniformly at random.
            \State Set $\sigma(x^*)=
                \begin{cases}
                    1 \text{\quad with probability\;} \tau(B_{\Phi_t}(x^*,R)) \\
                    0 \text{\quad with probability\;} 1-\tau(B_{\Phi_t}(x^*,R))
                \end{cases}$
            \State \begin{varwidth}[t]{\linewidth}
                Obtain $\Phi_{t+1}$ from $\Phi_{t}$ by: \\
                (i) remove $x^*$; \\
                (ii) for any clause having $x^*$ before (i), add $\sigma(x^*)$ to its right-hand-side value; \\
                (iii) remove all clauses that no longer contain any variable.
            \end{varwidth}
        \EndFor
        \State Output: the assignment $\sigma$.
    \end{algorithmic}
\end{algorithm}

For any $t \in [n]$, if the value $\tau(B_{\Phi_t}(x^*,R))$ given by the local rule $\tau$ in the $t$-th iteration is $1/2$, then we call that iteration a \textbf{free step}.
In a free step, the $\tau$-decimation algorithm simply assigns a uniformly random Boolean value to the selected variable.
On the contrary, if the value $\tau(B_{\Phi_t}(x^*,R))$ given by the local rule $\tau$ in the $t$-th iteration is either 0 or 1, then we call that iteration a \textbf{forced step}.
In a forced step, the $\tau$-decimation algorithm is forced to assign a particular value to the selected variable according to the value $\tau(B_{\Phi_t}(x^*,R))$.
To simplify our discussion, we introduce the following definitions for those $\tau$-decimation algorithms having certain numbers of free steps.
\begin{definition}
    For any $\delta \in [0,1]$, we say a $\tau$-decimation algorithm \texttt{DEC\textsubscript{$\tau$}} is \textbf{$\delta$-free} on the random $k$-XORSAT instance $\mathbf{\Phi} \sim \mathbf{\Phi}_k(n,rn)$ if
    w.h.p the $\tau$-decimation algorithm $\texttt{DEC}_\tau$ on input $\mathbf{\Phi}$ has at least $\delta n$ free steps.
\end{definition}

\begin{definition}
    For any $\delta \in [0,1]$, we say a $\tau$-decimation algorithm \texttt{DEC\textsubscript{$\tau$}} is \textbf{strictly $\delta$-free} on the random $k$-XORSAT instance $\mathbf{\Phi} \sim \mathbf{\Phi}_k(n,rn)$ if
    there exists $\delta' > \delta$ such that
    the $\tau$-decimation algorithm \texttt{DEC\textsubscript{$\tau$}} is $\delta'$-free on $\mathbf{\Phi}$.
\end{definition}

There are many choices for the local rules $\tau$.
The simplest one is the Unit Clause Propagation \texttt{UC}.
In each iteration, after selecting the unassigned variable $x^*$, \texttt{UC} checks whether there exists a unit clause (clause with one variable) on the variable $x^*$.
If yes, then \texttt{UC} sets $\tau(B_{\Phi_t}(x^*,R))$ to be the right-hand-side value of the unit clause, which can force the decimation algorithm to pick the suitable value to satisfy that clause.
In this case, this iteration is a forced step.
(If there are multiple unit clauses on the selected variable $x^*$, then only consider the one with the lowest index.)
If there is no unit clause on the selected variable $x^*$, then \texttt{UC} sets $\tau(B_{\Phi_t}(x^*,R))$ to $1/2$, which let the algorithm choose the assigned value randomly.
In this case, this iteration is a free step.
\begin{algorithm}[H]
    \caption{Unit Clause Propagation \texttt{UC}}
    \begin{algorithmic}[1]
        \State Input: \begin{varwidth}[t]{\linewidth}
            the selected variable $x^*$, and
            its local neighborhood $B_{\Phi_t}(x^*,2)$
        \end{varwidth}
        \If {there exists any unit clause on the variable $x^*$}
            \State Pick the unit clause $c$ on the variable $x^*$ (with the lowest index if having multiple such clauses).
            \State Output: the right-hand-side value of the clause $c$.
        \Else
            \State Output: $1/2$.
        \EndIf
    \end{algorithmic}
\end{algorithm}

\subsection{Message passing algorithms}

A new challenger to break the algorithmic threshold came out from statistical mechanics.
In experiments \cite{mezardAnalyticAlgorithmicSolution2002, gomesSatisfiedPhysics2002, braunsteinSurveyPropagationAlgorithm2005, ricci-tersenghiCavityMethodDecimated2009, krocMessagepassingLocalHeuristics2009}, the \emph{message passing algorithms} demonstrated their high efficiency on finding solutions of random $k$-SAT problem with the densities close to the satisfiability threshold.
Those algorithms include \emph{Belief Propagation Guided Decimation Algorithm} and \emph{Survey Propagation Guided Decimation Algorithm}, which are based on the insightful but non-rigorous \emph{cavity method} from statistical mechanics \cite{braunsteinSurveyPropagationAlgorithm2005, ricci-tersenghiCavityMethodDecimated2009}.
Unfortunately, several analyses showed that they do not outperform the best known algorithms for some problems.
Coja-Oghlan \cite{coja-oghlanBeliefPropagationGuided2017} showed that BP-guided Decimation fails to find solutions for random $k$-SAT w.h.p. for density above $\rho_0 2^k/k$ for a universal constant $\rho_0 > 0$, and thus does not outperform the best known algorithm from \cite{coja-oghlanBetterAlgorithmRandom2010}.
Hetterich \cite{hetterichAnalysingSurveyPropagation2016} also gave the same conclusion for SP-guided Decimation by showing that it fails w.h.p. for density above $(1+o_k(1))2^k \ln k / k$.

For random NAE-$k$-SAT, Gamarnik and Sudan \cite{gamarnikPerformanceSequentialLocal2017} showed that the balanced sequential local algorithms fail to find solutions for density above $(1+o_k(1)) 2^{k-1} \ln^2 k/k$ for sufficiently large $k$.
This means the algorithms do not outperform the best known algorithm, Unit Clause algorithm, which can find solutions w.h.p. for density up to $\rho\,2^{k-1}/k$ for some universal constant $\rho>0$ for sufficiently large $k$ \cite{achlioptasTwocoloringRandomHypergraphs2002}.
The framework of balanced sequential local algorithms also covers BP-guided Decimation and SP-guided Decimation with the number of message passing iterations is bounded by $O((\ln \ln n)^{O(1)})$.

In our work, we obtain an analogous result.
In Theorem \ref{thm:main-theorem}, we show that w.h.p. strictly $2\mu(k,r)$-free sequential local algorithms fails to solve the random $k$-XORSAT problem when the clause density exceeds the clustering threshold.
Then, in Theorem \ref{thm:main-theorem-exact-marginal}, we show that any sequential local algorithm with local rule that can compute the exact marginals are strictly $2\mu(k,r)$-free and thus fails to find a solution for random $k$-XORSAT problem.
This theorem covers the sequential local algorithms with Belief Propagation \texttt{BP} and Survey Propagation \texttt{SP} as local rules.

\subsection{Technique}
\label{sec:technique}

The works from \cite{achlioptasAlgorithmicBarriersPhase2008, achlioptasSolutionspaceGeometryRandom2011} demonstrated the clustering phenomenon for several random CSPs, and conjectured that it could be an obstruction of solving those problems.
\cite{gamarnikLimitsLocalAlgorithms2017, gamarnikPerformanceSequentialLocal2017} and subsequent works leveraged a different notion of clustering, named \emph{overlap gap property} (OGP) by \cite{gamarnikFindingLargeSubmatrix2018}, to link the clustering phenomenon to the hardness rigorously.
Gamarnik gave a detailed survey on it \cite{gamarnikOverlapGapProperty2021}.

This paper focuses on the vanilla version of the OGP.
Given an instance $\Phi$ of the constraint satisfaction problem, we say it exhibits the overlap gap property with values $0 \le v_1 < v_2$ if every two solutions $\sigma$ and $\sigma'$ satisfy either $d(\sigma,\sigma') \le v_1$ or $d(\sigma,\sigma') \ge v_2$, where $d$ is a metric on its solution space.
(We assume $d$ is the Hamming distance throughout this paper.)
Intuitively, it means every pair of solutions are either close to each other, or far from each other, and thus the solution space of the instance exhibits a topological discontinuity based on the proximity.

Now we illustrate how does the overlap gap property method.
(Some details of the overlap gap property method is slightly different if we consider different variants of the OGP, but the overall idea of \emph{topological barrier} stays the same.)
Assume we have an algorithm $\mathcal{A}$ that takes a random instance $\mathbf{\Phi}$ as input and outputs an assignment $\sigma$ for the instance.
The output $\sigma$ can be viewed as a random variable which depends on both the random instance $\mathbf{\Phi}$ and some internal random variables, represented by a random internal vector $\mathbf{I}$, of the algorithm.
Let $\Phi$ and $I$ be realizations of $\mathbf{\Phi}$ and $\mathbf{I}$ respectively, and denote the output assignment as $\sigma_0=\sigma_{\Phi,I}$.
We then re-randomize the components of the internal vector $I$ \emph{one-by-one}.
After each re-randomizing a component of $I$, we run the algorithm again to generate a new output assignment.
Then, we obtain a sequence of assignments $\sigma_0, \sigma_1, \cdots, \sigma_T$ for the instance $\Phi$, where $T$ is the number of components of the random vector $I$ that we have re-randomized.
Next, we show that the algorithm is \emph{insensitive to its input} in the sense that when one of the components of the random internal vector $\mathbf{I}$ is re-randomized, the output assignment almost remains unchanged,
In particular, $d(\sigma_{i}, \sigma_{i+1})$ is smaller than $v_2 - v_1$.
We also show that the algorithm has certain \emph{freeness} in the sense that when all components in the random internal are re-randomized the output assignment is expected to change a lot.
In particular, $\expected{d(\sigma_0, \sigma_T)}$ should be larger than $v_2$.
These two properties together imply that the sequence of assignments cannot "jump" over the overlap gap, while two ends of the sequence probably lie in different clusters.
Therefore, there should be an assignment $\sigma_{T_0}$ that falls in the gap, namely, there exists $T_0>0$ such that $v_1 \le d(\sigma_0,\sigma_{T_0}) \le v_2$.
If the probability that the algorithm successfully finds a solution is greater than some small value $s_n$ slowly converging to 0, then there could be a very small probability that both $\sigma_0$ and $\sigma_{T_0}$ are solutions of the instance $\Phi$, with $v_1 \le d(\sigma_0,\sigma_{T_0}) \le v_2$.
Even though this probability is very small, it still has the chance to violates the OGP of the instance.
Then, by contradiction, we could conclude that the probability that the algorithm succeeds in finding a solution is smaller than $s_n=o(1)$, namely, w.h.p. the algorithm fails in finding a solution.

Instead of considering the overlap gap property of the entire instance $\mathbf{\Phi}$, we move our focus to the overlap gap property of a sub-instance of $\mathbf{\Phi}$.
Indeed, the sub-instance we consider is the \emph{2-core instance} $\mathbf{\Phi_c}$ induced by the 2-core of the factor graph representation of the random instance $\mathbf{\Phi}$.
In \cite{ibrahimiSetSolutionsRandom2012, achlioptasSolutionSpaceGeometry2015}, they proved that the 2-core instance $\mathbf{\Phi_c}$ exhibits the overlap gap property with $v'_1=o(n)$ and $v'_2=\epsilon_k n$ for some constant $\epsilon_k>0$ for clause density $r_{core}(k) < r < r_{sat}(k)$.
We remove all variables not in the core instance from the sequence of assignments $\sigma_0, \sigma_1, \cdots, \sigma_T$ we obtained above, then it becomes a sequence of assignments $\sigma'_0, \sigma'_1, \cdots, \sigma'_T$ for the core instance.
We also prove that the algorithm is \emph{insensitive to its input} with respect to the core instance in the sense that $d(\sigma'_{i},\sigma'_{i+1}) < v'_2 - v'_1$, and has certain \emph{freeness} so that $\expected{d(\sigma'_0,\sigma'_T)} > v'_2$.
By repeating the above argument of the overlap gap property method, we can conclude that w.h.p. the algorithm fails in find a solution.

Our proof can also be used for the non-sequential local algorithms.
Since the local rule $\tau$ runs on the local neighborhood of each variable in parallel, the values assigned to variables do not depend on each other.
Informally speaking, there is no long-range dependency among those assigned values.
Therefore, re-randomizing one component of the internal vector $I$, say $I_{i+1}$, only affects the value $\sigma(x_{i+1})$ assigned to the corresponding variable $x_{i+1}$.
So, we have $d(\sigma_i, \sigma_{i+1}) \le 1 < v'_2 - v'_1$.
Hence, we can obtain the same result as Theorem \ref{thm:main-theorem} for non-sequential local algorithms with the same proof.

\subsection{Related works}

The vanilla version of OGP helps us rule out some large classes of algorithms on random CSPs for relatively high clause densities, but it is not sophisticated enough to close the \emph{statistical-to-computational gap} in some cases such as random NAE-$k$-SAT discussed in \cite{gamarnikPerformanceSequentialLocal2017} and random $k$-XORSAT discussed in this paper (more details in Section \ref{sec:overlap-gap-property}).
There have been some recent works trying to improve the notion of OGP by developing different variants of OGP.
The most notable one is \emph{multi-OGP} \cite{weinOptimalLowdegreeHardness2022, breslerAlgorithmicPhaseTransition2022, huangTightLipschitzHardness2022, gamarnikGeometricBarriersStable2023, gamarnikAlgorithmicObstructionsRandom2023}, which succeeds in closing the statistical-to-computational gap in certain models.
However, it is not clear about the relation between the clustering property and the multi-OGP.

\section{Overlap Gap Property}
\label{sec:overlap-gap-property}


We say that a $k$-XORSAT instance $\Phi$ exhibits the \textbf{overlap gap property} (or shortened as OGP) with the values $0 \le v_1 < v_2$ if for any two solutions $\sigma, \sigma' \in S(\Phi)$ we have either $d(\sigma,\sigma') \le v_1$ or $d(\sigma,\sigma') \ge v_2$.
Informal speaking, any two solutions of the instance are either close to each other or far away from each other, and thus the solution space exhibits a topological discontinuity.
Given a random $k$-XORSAT instance, we can prove that it exhibits the OGP when the clause density is greater than certain value, and obtain the following lemma.
\ifappendix
(The proof is given in the Appendix \ref{sec:proof-of-ogp-of-whole-instance}.)
\fi

\begin{lemma}
    \label{lemma:r-1}
    For any $k \ge 3$, there exists $r_1(k) > 0$ such that for $r > r_1(k)$ and any pair of solutions $\sigma,\sigma' \in S(\mathbf{\Phi})$ of the random $k$-XORSAT instance $\mathbf{\Phi} \sim \mathbf{\Phi}_n(k,rn)$, w.h.p. the distance $d(\sigma,\sigma')$ between the two solutions is either $\le u_1n$ or $\ge u_2n$ for some $0 \le u_1 < u_2$.
    In particular, the value of $r_1(k)$ is given by
    \begin{align*}
        r_1(k) = \min_{0 \le \alpha \le \frac{1}{2}} \frac{1 + H(\alpha)}{ 2 - \log_2 (1+(1-2\alpha)^k)},
    \end{align*}
    where $H$ is the binary entropy function, that is, $H(x) = -x\log_2(x) - (1-x)\log_2(1-x)$.
\end{lemma}


Instead of considering the random $k$-XORSAT instance $\mathbf{\Phi}$ itself, we focus on the sub-instance, called the \emph{core instance} (defined below), of the random $k$-XORSAT instance $\mathbf{\Phi}$, and show that core instance also exhibits the overlap gap property, even when the clause density is much lower. (See Table \ref{tab:r-core-vs-r-1}.)

\begin{table}[ht]
    \centering
    \begin{tabular}{|c|c|c|c|c|c|c|c|c|}
        \hline
        $k$     & 3 & 4 & 5 & 6 & 7 & 8 & 9 \\
        \hline
        $r_{core}(k)$ &
        0.818470 &
        0.772280 &
        0.701780 &
        0.637081 &
        0.581775 &
        0.534997 &
        0.495255 \\
        \hline
        $r_1(k)$ &
        0.984516 &
        0.943723 &
        0.905812 &
        0.874349 &
        0.848314 &
        0.826470 &
        0.807862 \\
        \hline
    \end{tabular}
    \caption{Compare $r_{core}(k)$ with $r_1(k)$ for different $k$. The numeric values in the table are rounded off to 6 decimal places.}
    \label{tab:r-core-vs-r-1}
\end{table}

We start from defining the \textbf{peeling algorithm} and the \textbf{core instances}.
Given an XORSAT instance $\Phi$, suppose there exists a variable $x$ of degree 1, which means it is involved in exactly one clause $e$.
We remove the variable $x$ and the only clause $c$ involving $x$, to obtain a modified instance $\Phi'$.
If we have a solution $\sigma'$ for the modified instance $\Phi'$, we can always choose a suitable value for the variable $x$ to satisfy the clause $c$, and extend the solution $\sigma'$ to a solution $\sigma$ for the original instance $\Phi$.
Similarly, we can also do the same thing if the variable $x$ is of degree 0, since it does not involve in any equation, and we are free to choose any value for it.
By doing this, solving the original instance $\Phi$ is reduced to solving the modified instance $\Phi'$.

We can repeat this process until there is no variable of degree at most 1.
This process is named the \textbf{peeling algorithm} on an instance as its input (Algorithm \ref{algo:peeling_algorithm_instance}).
We call the resultant instance the \textbf{2-core instance} (or simply the \textbf{core instance}) of the instance $\Phi$, denoted by $\Phi_c$.
This name is borrowed from the graph theory.
In graph theory, the \emph{$k$-core} of a graph is the maximal subgraph with minimum degree at least $k$.
It is known that the factor graph of the core instance $\Phi_c$ is exactly the maximal subgraph of the factor graph $G_\Phi$ of the instance $\Phi$, with minimum variable degree at least 2.

\begin{algorithm}
    \caption{Peeling algorithm}
    \label{algo:peeling_algorithm_instance}
    \begin{algorithmic}[1]
        \State Input: an instance $\Phi$.
        \While {There exists $\ge 1$ variable of degree $\le 1$.}
            \State Select a variable $x$ of degree $\le 1$.

            (Pick $x$ with the lowest index if there are $> 1$ such variables.)
            \State Update $\Phi$ by removing the variable $x_i$ and its only involved clause (if exists).
        \EndWhile
        \State Output: the resultant instance $\Phi$.
    \end{algorithmic}
\end{algorithm}


Mézard and Montanari \cite{mezardInformationPhysicsComputation2009} gave a detailed description on the structure of the core instance.
Reader can find more details about the core instance in their book.
The following theorem is a short summary of some known facts about the core instances we needed in the paper.

\begin{theorem}
    \label{thm:core_emerges}
    For any $k \ge 3$, there exists $r_{core}(k) > 0$ given by
    \begin{align*}
        r_{core}(k) \equiv \sup \{ r\in[0,1]: Q > 1-e^{-krQ^{k-1}} \; \forall Q\in(0,1) \}
    \end{align*}
    such that the factor graph $G_c$ of the core instance $\mathbf{\Phi}_c$ of the random $k$-XORSAT instance $\mathbf{\Phi} \sim \mathbf{\Phi}_n(k,rn)$ have the following properties.
    \begin{enumerate}
        \item For $r<r_{core}(k)$, w.h.p. the factor graph $G_c$ of the core instance $\mathbf{\Phi}_c$ is an empty graph.
        \item For $r>r_{core}(k)$, w.h.p. the factor graph $G_c$ of the core instance $\mathbf{\Phi}_c$ have $V(k,r)n+o(n)$ variable nodes, where
        \begin{align*}
            V(k,r) &= 1 - \exp(-krQ^{k-1}) - krQ^{k-1} \exp(-krQ^{k-1})
        \end{align*}
            and $Q$ is the largest solution of the fixed point equation $Q=1-\exp(-krQ^{k-1})$.
        In particular, the fraction of variable nodes of degree $l$ is between $\widehat{\Lambda}_l-\epsilon$ and $\widehat{\Lambda}_l+\epsilon$ with probability greater than $1-e^{-\Theta(n)}$, where $\widehat{\Lambda}_0 = \widehat{\Lambda}_1 = 0$ and
        \begin{align*}
            \widehat{\Lambda}_l &= \frac{1}{e^{krQ^{k-1}}-1-krQ^{k-1}} \; \frac{1}{l!} \; (krQ^{k-1})^l \text{\quad for \;} l \ge 2.
        \end{align*}
    \item Conditioning on the number of variable nodes $V(k,r)n+o(n)$ and the degree profile $\widehat{\Lambda}$, the factor graph $G_c$ of the core instance $\mathbf{\Phi}_c$ is distributed according to the ensemble containing all possible factor graphs of $k$-XORSAT instances of $V(k,r)n+o(n)$ variables and variable degree distribution $\Lambda$.
    \end{enumerate}
\end{theorem}

Theorem \ref{thm:core_emerges} shows that there exists a threshold $r_{core}(k)$ below the satisfiability threshold $r_{sat}(k)$ of random $k$-XORSAT problem.
When the clause density $r$ is below the threshold $r_{core}(k)$, w.h.p. the instance $\mathbf{\Phi}$ does not have a core instance.
When the clause density $r$ is above the threshold $r_{core}(k)$, w.h.p. the core instance $\mathbf{\Phi}_c$ emerges.
In particular, the variable degree distribution is a Poisson distribution with mean $krQ^{k-1}$ conditioning on $\Lambda_0 = \Lambda_1 = 0$.
\cite{mezardInformationPhysicsComputation2009} also showed that the core instance exhibits the OGP.

\begin{lemma}
    \label{lemma:ogp_of_core}
    For $k \ge 3$ and $r_{core}(k) < r < r_{sat}(k)$, there exists $\epsilon(k,r) > 0 $ such that w.h.p. the distance between any two solutions for the core instance $\mathbf{\Phi}_c$ of a random $k$-XORSAT instance $\mathbf{\Phi} \sim \mathbf{\Phi}_n(k,rn)$ is either $o(n)$ or greater than $\epsilon(k,r) n$.
\end{lemma}

Now, we know that w.h.p. the core instance of a random $k$-XORSAT instance has the overlap gap property with the values $v_1 = o(n)$ and $v_2 = \epsilon(k,r)n$.
With OGP, we can partition the solution space of the core instance into multiple groups, each called a \textbf{core cluster}, such that the distance between any pair of core solutions in the same core cluster is at most $o(n)$, and the distance between any pair of core solutions in different core clusters is at least $\epsilon(k,r) n$.

Suppose we have a $k$-XORSAT instance $\Phi$.
We first define a binary relation on the solution space of a core instance $\Phi_c$ by: for $\sigma_c, \sigma'_c \in \mathcal{S}(\Phi_c)$, we write $\sigma_c \simeq \sigma'_c$ if and only if $d(\sigma_c, \sigma'_c) = o(n)$.
It is easy to see it is an equivalence relation.
Then, we can partition the solution space by the equivalence classes of $\simeq$.
We can denote those equivalence classes by $\mathcal{S}_{c,1}, \mathcal{S}_{c,2}, ..., \mathcal{S}_{c,n_c}$.
Thus, we have $\mathcal{S}_{c,1} \sqcup \mathcal{S}_{c,2} \sqcup ...\sqcup \mathcal{S}_{c,n_c} = \mathcal{S}(\Phi_c)$ where $\sqcup$ is the disjoint union.
Then, we have
\begin{align*}
    d(\sigma_c, \sigma'_c) &= o(n)
    \text{\quad\quad\quad\, if \;} \sigma_c, \sigma'_c \in \mathcal{S}_{c,i}, \text{\quad and}  \\
    d(\sigma_c, \sigma'_c) &\ge \epsilon(k,r) n 
    \text{\quad\quad if \;} \sigma_c \in \mathcal{S}_{c,i}, \;\sigma'_c \in \mathcal{S}_{c,j} \text{\;and\;} \mathcal{S}_{c,i} \neq \mathcal{S}_{c,j}
\end{align*}

Now we can partition the solution space $\mathcal{S}(\Phi)$ of the original instance $\Phi$ into clusters based on the partition of the solution space of core instance.
We set
\begin{align*}
    \mathcal{S}(\Phi) = \bigsqcup_{i=1}^{n_c} \mathcal{S}_i
    \text{\quad and \quad}
    \mathcal{S}_i = \{\sigma \in \mathcal{S}(\Phi):\pi(\sigma) \in \mathcal{S}_{c,i}\} \text{\quad for\;} i=1,2,...,n_c
\end{align*}
where $\pi$ is defined to be the \textbf{projection} mapping assignments for the instance $\Phi$ to assignments for the core instance $\Phi_c$ by removing all variables not in the core instance $\Phi_c$.
Each $\mathcal{S}_i$ is called a \textbf{cluster} in the solutions space $\mathcal{S}(\Phi)$.
We can then prove that these clusters are well-separated from each other.

\begin{lemma}
    \label{lemma:partition-sol-space}
    Let $k \ge 3$ and $r_{core}(k) < r < r_{sat}(k)$.
    Suppose $\mathbf{\Phi} \sim \mathbf{\Phi}_n(k,rn)$ is a random $k$-XORSAT instance.
    Then, w.h.p. there exists a partition $\mathcal{S}(\mathbf{\Phi}) = \mathcal{S}_1 \sqcup \mathcal{S}_1 \sqcup ...\sqcup \mathcal{S}_{n_c}$ for the solutions space $S(\mathbf{\Phi})$ of the random instance $\mathbf{\Phi}$ such that the following statements hold.
    \begin{enumerate}
        \item If $\sigma, \sigma' \in \mathcal{S}_i$ for some $i \in [n_c]$, then we have $d(\sigma, \sigma') \le \mu(k,r)n + o(n)$,
        where the real-valued function $\mu(k,r)$ is given by
        $\mu(k,r) = \exp(-krQ^{k-1})+krQ^{k-1}\exp(-krQ^{k-1})$
        and $Q$ is the largest solution of the fixed point equation $Q=1-\exp(-krQ^{k-1})$.
        \item If $\sigma \in \mathcal{S}_i, \sigma' \in \mathcal{S}_j$ and $\mathcal{S}_i \neq \mathcal{S}_j$ for some $i,j \in [n_c]$, then we have $d(\sigma, \sigma') \ge \epsilon(k,r)n$.
    \end{enumerate}
\end{lemma}

\begin{proof}
    Assume the instance $\mathbf{\Phi}$ has a non-empty core instance $\mathbf{\Phi}_c$, which exists with high probability according to Theorem \ref{thm:core_emerges}.
    We also assume the core instance $\mathbf{\Phi}_c$ exhibits the OGP with $v_1=o(n)$ and $v_2=\epsilon(k,r)n$, which occurs with high probability according to Lemma \ref{lemma:ogp_of_core}.
    Let $\sigma$ and $\sigma'$ be two solutions of the random $k$-XORSAT instance $\mathbf{\Phi}$, and let $\sigma_c = \pi(\sigma)$ and $\sigma'_c = \pi(\sigma')$ be the projection of $\sigma$ and $\sigma'$ on the core solution space $\mathcal{S}(\mathbf{\Phi})$, respectively.

    To prove the first part of the lemma, we assume that $\sigma$ and $\sigma'$ are in the same cluster, that is, $\sigma, \sigma' \in \mathcal{S}_i$ for some $i \in [n_c]$.
    By the definition of cluster, we have $d(\sigma_c,\sigma'_c) = o(n)$.
    Therefore, $d(\sigma,\tau)$ is upper bounded by the number of variables not in the core instance, plus $o(n)$.
    By Theorem \ref{thm:core_emerges}, the number of variables outside the core instance is given by $(1-V(k,r))n+o(n)$.
    Hence, we have
        $d(\sigma, \tau)
        \le (1-V(k,r))n + o(n)
        = \left( \exp(-krQ^{k-1})+krQ^{k-1}\exp(-krQ^{k-1}) \right)n + o(n)$.

    To prove the second part of the lemma, we assume that $\sigma$ and $\tau$ are in the different clusters, that is, $\sigma \in \mathcal{S}_i$, $\sigma' \in \mathcal{S}_j$ and $\mathcal{S}_i \neq \mathcal{S}_j$ for some $i,j \in [n_c]$.
    By the definition of cluster and Lemma \ref{lemma:ogp_of_core}, we have $d(\sigma_c,\sigma'_c) \ge \epsilon(k,r)n$
    Therefore, we have $d(\sigma, \sigma') \ge d(\pi(\sigma),\pi(\sigma')) = d(\sigma_c,\sigma'_c) \ge \epsilon(k,r)n$.
\end{proof}

\section{Preparation of OGP method}

In this section, we introduce some notions and obtain some preliminary results needed by the overlap gap property method to prove the main results.

\subsection{Sequence of output assignments}
\label{sec:sequence-of-output-assignments}

The random $k$-XORSAT instance $\mathbf{\Phi}$ is a random variable, and the $\tau$-decimation algorithm \texttt{DEC$_\tau$} is a randomized algorithm.
Therefore, the assignment output by the $\tau$-decimation algorithm \texttt{DEC$_\tau$} on input $\mathbf{\Phi}$ is also a random variable.
The outcomes of the output assignment depend on the random instance $\mathbf{\Phi}$, the order of variables being chosen, and the value selection based on the output from the local rule $\tau$.
Now we introduce two random variables to explicitly represent the order of variables and the value selection so that we can have a more concrete language to discuss how the randomness from both the instance and the algorithm affects the output assignment.
We adopt the notation from \cite{gamarnikPerformanceSequentialLocal2017} in the following discussion.

The order of variables can be represented by a random vector $\mathbf{Z} = (\mathbf{Z}_1, \mathbf{Z}_2, \cdots, \mathbf{Z}_n)$ whose entries are $n$ i.i.d. random variables with uniform distribution over the interval $[0,1] \subset \mathbb{R}$, independent of the random instance $\mathbf{\Phi}$.
We call $\mathbf{Z}$ the \textbf{ordering vector} of the algorithm.
For all $i \in [n]$, the variable $x_i$ in the instance $\mathbf{\Phi}$ is associated with the random variable $\mathbf{Z}_i$.
In each iteration of the algorithm, the unassigned variable $x_i$ with the largest value $\mathbf{Z}_i$, among all other unassigned variables, is selected.
In the other words, we can construct the permutation $s:[n] \rightarrow [n]$ such that $\mathbf{Z}_{s(1)} > \mathbf{Z}_{s(2)} > \cdots > \mathbf{Z}_{s(n)}$, and for all $t \in [n]$ the variable $x_{s(t)}$ is selected in the $t$-th iteration.
The value selection based the output from the local rule $\tau$ can be represented by a random vector $\mathbf{U} = (\mathbf{U}_1, \mathbf{U}_2, ..., \mathbf{U}_n)$ whose entries are $n$ i.i.d. random variables with uniform distribution over the interval $[0,1] \subset \mathbb{R}$.
We call $\mathbf{U}$ the \textbf{internal vector} of the algorithm.
In the $t$-th iteration of the algorithm, the value $\sigma(x_{s(t)})$ assigned to the selected variable $x_{s(t)}$ is set to be 1 if $\mathbf{U}_t < \tau(B_{\mathbf{\Phi}_t}(x_{s(t)},R))$, and 0 otherwise.
Conditioning on $\mathbf{\Phi}$, $\mathbf{Z}$ and $\mathbf{U}$, the output assignment $\sigma$ can be uniquely determined.
Therefore, we can view the $\tau$-decimation algorithm \texttt{DEC$_\tau$} as a deterministic algorithm on random input $(\mathbf{\Phi}, \mathbf{Z}, \mathbf{U})$, and denote by $\sigma_{\mathbf{\Phi}, \mathbf{Z}, \mathbf{U}}$ the output of the algorithm.

With this notion of the deterministic algorithm, we can construct a sequence of output assignments which will be used in the argument of the overlap gap property method.
The sequence of output assignments is generated by applying the $\tau$-decimation algorithm \texttt{DEC$_\tau$} on a random $k$-XORSAT instance $\mathbf{\Phi}$ multiple times in the following way:
First, given a random $k$-XORSAT instance $\mathbf{\Phi}$, we sample an ordering vector $\mathbf{Z}$ and an internal vector $\mathbf{U}$.
Then, we run the $\tau$-decimation algorithm \texttt{DEC$_\tau$} on input $\mathbf{\Phi}$ with the ordering vector $\mathbf{Z}$ and the internal vector $\mathbf{U}$ to get the first output assignment $\sigma_0$.
After that, we re-randomize (i.e. sample again) the entries of the internal vector $\mathbf{U}$ one by one from $\mathbf{U}_1$ to $\mathbf{U}_n$.
Right after each re-randomization we run the algorithm again to get a new output assignment.
By doing this, we obtain a sequence of $n+1$ output assignments for the instance $\mathbf{\Phi}$ in total.
We denote by $\sigma_i$ the output assignment generated after re-randomizing the first $i$ entries of $\mathbf{U}$, for $i=0,1,2,...,n$.
Precisely speaking, let $\mathbf{V}=(\mathbf{V}_1, \mathbf{V}_2, ..., \mathbf{V}_n)$ and $\mathbf{W}=(\mathbf{W}_1, \mathbf{W}_2, ..., \mathbf{W}_n)$ be two independent random internal vectors with the uniform distribution over $[0,1]^n$, and set $\mathbf{U^i} = (\mathbf{W}_1, ..., \mathbf{W}_i, \mathbf{V}_{i+1}, ..., \mathbf{V}_n)$ for $i=0,1,2,...,n$.
Note that $\mathbf{U^0}=\mathbf{V}$ and $\mathbf{U^n}=\mathbf{W}$.
Then, the sequence of output assignments $\{\sigma_i\}_{i=0}^n$ can be written as $\{\sigma_{\mathbf{\Phi},\mathbf{Z},\mathbf{U^i}}\}_{i=0}^n$, which is equivalent to the sequence of output assignment obtained by running the $\tau$-decimation algorithm \texttt{DEC$_\tau$} (for $n+1$ times in total) on input $(\mathbf{\Phi},\mathbf{Z},\mathbf{U^i})$ for all $i=0,1,...,n$.

Recall the projection $\pi$ mapping assignments for the instance $\mathbf{\Phi}$ to assignments for the core instance $\mathbf{\Phi}_c$, by removing all variables not in the core instance.
We can further obtain a sequence of assignments for the core instance $\mathbf{\Phi}_c$ by applying the projection on the output assignments $\sigma_i$, that is, we set $\{\sigma'_i = \pi(\sigma_{\mathbf{\Phi},\mathbf{Z},\mathbf{U^i}})\}_{i=0}^n$.

\subsection{Insensitive to internal vector}

In this section, we show that the $\tau$-decimation algorithm \texttt{DEC$_\tau$} is \emph{insensitive} to its internal vector.
By \emph{insensitive}, it means when the value of an entry in the internal vector $\mathbf{U}$ is changed, only a small portion of the assigned values in the output assignment $\sigma_{\mathbf{\Phi},\mathbf{Z},\mathbf{U}}$ change accordingly.
If so, every two consecutive output assignments in the sequence $\{\sigma_i=\sigma_{\Phi,\mathbf{Z},\mathbf{U^i}}\}_{i=0}^n$ should only differ from each other in only a small portion of assigned values.

Consider the sequence of output assignment $\{\sigma_i = \sigma_{\mathbf{\Phi},\mathbf{Z},\mathbf{U^i}}\}_{i=0}^n$ from Section \ref{sec:sequence-of-output-assignments}.
Note that the $i$-th output assignment $\sigma_i$ in the sequence is the output of the algorithm on input $(\mathbf{\Phi},\mathbf{Z},\mathbf{U^i})$.
For any $i \in [n]$, the only difference between the input $(\mathbf{\Phi},\mathbf{Z},\mathbf{U^{i-1}})$ and the input $(\mathbf{\Phi},\mathbf{Z},\mathbf{U^{i}})$ is the $i$-th entries of the internal vectors $\mathbf{U^{i-1}}$ and $\mathbf{U^{i}}$.
We can immediately see that the insensitivity of the algorithm implies that every two consecutive output assignments in the sequence are close to each other.
Gamarnik and Sudan \cite{gamarnikPerformanceSequentialLocal2017} proved the insensitivity of the $\tau$-decimation algorithm in their works, using the notion of \emph{influence range}.
Although their works \cite{gamarnikPerformanceSequentialLocal2017} focused on the random NAE-$k$-SAT problem, the proof for the insensitivity of the $\tau$-decimation algorithm is independent of the type of clauses in the random constraint satisfaction framework.
So, we can directly use the result here.

\begin{definition}
    Given a random instance $\mathbf{\Phi}$ and a random ordering vector $\mathbf{Z}$, we say that $x_i$ \textbf{influences} $x_j$ if either $x_i=x_j$ or in the variable-to-variable graph of the instance $\mathbf{\Phi}$ there exists a sequence of variable nodes $y_0, y_1, ..., y_t \in \{x_1, x_2, ..., x_n\}$ such that the following statements hold.
    \begin{enumerate}
        \item $y_0=x_i$ and $y_t=x_j$.
        \item There exists a path from $y_l$ to $y_{l+1}$, of length at most $r$, in the variable-to-variable graph $G$, for $l=0,1,...,t-1$.
        \item $\mathbf{Z}_{y_{l-1}} > \mathbf{Z}_{y_{l}}$ for $l=1,2,...,t$. In particular, $\mathbf{Z}_{x_i} > \mathbf{Z}_{x_j}$.
    \end{enumerate}
    We define the \textbf{influence range} of $x_i$ to be the set of all variables $x_j$ influenced by $x_i$, denoted by $\mathcal{IR}_{x_i}$.
\end{definition}

\begin{lemma}
    \label{lemma:influence-range}
    Given an instance $\Phi$, a vector $Z \in [0,1]^n$, and two vectors $U,U' \in [0,1]^n$, we assume there exists $i \in \{1,2,...,n\}$ such that $U_{i} \neq U'_{i}$ and $U_{j} = U'_{j}$ for all $j \neq i$.
    Then, $\sigma_{\Phi,Z,U}(x) = \sigma_{\Phi,Z,U'}(x)$ for all variables $x_j \notin \mathcal{IR}_{x_{i}}$.
\end{lemma}

\begin{lemma}
    \label{lemma:influence-range-size}
    For any $\xi \in (0,1)$ and sufficiently large $n$,
    \begin{equation*}
        \proba{ \max_{1\le i \le n} |\mathcal{IR}_{x_i}| \ge n^{1/6} } \le \exp \left( -\ln n(\ln \ln n)^{\xi/4} \right).
    \end{equation*}
\end{lemma}

They first showed that changing the value of only one entry, say $U_i$, in the internal vector $U$ only affects the values assigned to the variables in the influence range of the variable $x_i$ (Lemma \ref{lemma:influence-range}).
They further showed that w.h.p. the size of the influence range of variables is sublinear for all variables (Lemma \ref{lemma:influence-range-size}).
Note that in the original statement of Lemma \ref{lemma:influence-range-size} in \cite{gamarnikPerformanceSequentialLocal2017}, the index $1/6$ in the inequality above can be any real number between 0 and 1/5.
Here, we pick a fixed value $1/6$ for simplicity.
Combining these two lemmas, we can show that w.h.p. the differences between $\sigma_{\mathbf{\Phi},\mathbf{Z},\mathbf{U^{i-1}}}$ and $\sigma_{\mathbf{\Phi},\mathbf{Z},\mathbf{U^{i}}}$ is upper bounded by $n^{1/6}$ for all $i \in [n]$.

\begin{lemma}
    \label{lemma:insensitive}
    For any $\xi \in (0,1)$ and sufficiently large $n$, 
    \begin{equation*}
        \proba{ d(\sigma_{\mathbf{\Phi},\mathbf{Z},\mathbf{U^{i-1}}}, \sigma_{\mathbf{\Phi},\mathbf{Z},\mathbf{U^{i}}}) \ge n^{1/6} \textnormal{\;\;for some\;} i \in [n]} \le \exp \left( -\ln n(\ln \ln n)^{\xi/4} \right).
    \end{equation*}
\end{lemma}

\begin{proof}
    Fix an arbitrary $i \in [n]$.
    We know that $\mathbf{U}_j^\mathbf{i-1} = \mathbf{U}_j^\mathbf{i}$ for all $j \neq i$, and $\mathbf{U}_i^\mathbf{i-1} \neq \mathbf{U}_i^\mathbf{i}$.
    By Lemma \ref{lemma:influence-range}, we have $\sigma_{\mathbf{\Phi},\mathbf{Z},\mathbf{U^{i-1}}}(x_j) = \sigma_{\mathbf{\Phi},\mathbf{Z},\mathbf{U^{i}}}(x_j)$ for all variables $x_j \notin \mathcal{IR}_{x_i}$.
    If $d(\sigma_{\mathbf{\Phi},\mathbf{Z},\mathbf{U^{i-1}}}, \sigma_{\mathbf{\Phi},\mathbf{Z},\mathbf{U^{i}}}) \ge n^{1/6}$ for some $i \in [n]$, we have $|\mathcal{IR}_{x_i}| \ge n^{1/6}$.
    Hence, by Lemma \ref{lemma:influence-range-size}, the result follows.
\end{proof}

\subsection{Freeness}
Recall the definition of \emph{free steps}.
An iteration of the $\tau$-decimation algorithm \texttt{DEC$_\tau$} is called a free step if the local rule $\tau$ gives the value $1/2$ in that iteration.
In this case, the value chosen by the $\tau$-decimation algorithm for the selected variable is either 0 or 1 with even probability.
Intuitively, it means that the local rule $\tau$ cannot capture useful information from the local structure to guide the $\tau$-decimation algorithm choosing value for the selected variable, and thus the $\tau$-decimation algorithm simply make a random guess for the assigned value.
We also recall the definition of a $\tau$-decimation algorithm being $\delta$-free.
A $\tau$-decimation algorithm \texttt{DEC$_\tau$} is $\delta$-free on the random $k$-XORSAT instance $\mathbf{\Phi}$ if w.h.p. the algorithm has at least $\delta n$ free steps, on input $\mathbf{\Phi}$.
Informal speaking, the more free the $\tau$-decimation algorithm, the less the information captured by the local rule.

By using the Wormald's method of differential equations, we can calculate the degree profile of the remaining factor graph after $t$ steps of the $\tau$-decimation algorithm, for all $0 \le t \le n$.
With the degree profiles, we can calculate the probability of each step being free, and thus approximate how free the $\tau$-decimation algorithm is.
The probability of having free steps depends on the choice of the local rules.
Lemma \ref{lemma:uc-free} shows the freeness of the \texttt{UC}-decimation algorithm.

\begin{lemma}
    \label{lemma:uc-free}
    For $k \ge 3$ and $r > 0$, the \textnormal{\texttt{UC}}-decimation algorithm \textnormal{\texttt{DEC\textsubscript{UC}}} is $w_1(k,r)$-free on the random $k$-XORSAT instance $\mathbf{\Phi} \sim \mathbf{\Phi}_k(n,rn)$, where
    \begin{align*}
        w_1(k,r) = \frac{(kr)^{\frac{1}{1-k}} }{k-1} \; \gamma \left( \frac{1}{k-1}, kr \right)
    \end{align*}
    and $\gamma$ is the lower incomplete gamma function given by
    $\gamma(a,x) \equiv \int_{0}^{x} t^{a-1} e^{-t} dt$.
\end{lemma}

The role of the local rules is to approximate the marginal probability of the selected variable over a randomly chosen solution for the sub-instance induced by the local neighborhood of the selected variable.
Interestingly, even we have a local rule $\tau$ that is capable to give the exact marginals when the factor graph is a tree, it still cannot provide enough useful information to guide the $\tau$-decimation algorithm making good decision for the assigned value.
With such a local rule, the $\tau$-decimation algorithm still has a certain level of freeness.
\begin{lemma}
    \label{lemma:bp-free}
    Assume the local rule $\tau$ can give the exact marginal probabilities of variables on any factor graph that is a tree.
    For $k \ge 3$ and $r > 0$, the $\tau$-decimation algorithm \textnormal{\texttt{DEC$_\tau$}} is $w_e(k,r)$-free on the random $k$-XORSAT instance $\mathbf{\Phi} \sim \mathbf{\Phi}_k(n,rn)$, where
    \begin{align*}
        w_e(k,r) = \int_{0}^{1} S_R(x) dx,
    \end{align*}
    $S_0(x) = 1$ and $S_{l}(x) = \exp \left( -kr [ (1-x)(1-S_{l-1}(x)) + x ]^{k-1} \right)$ for any $l \ge 1$ and $x \in \mathbb{R}$.
\end{lemma}

\ifappendix
The proof for Lemma \ref{lemma:uc-free} and Lemma \ref{lemma:bp-free} can be found in Appendix \ref{sec:uc-free-steps} and Appendix \ref{sec:bp-free-steps}, respectively.
\fi

\section{Proof of main theorems}
\label{sec:proof-of-main-theorems}

We denote by $\alpha_n$ the \emph{success probability} of the $\tau$-decimation algorithm \texttt{DEC\textsubscript{$\tau$}}, namely, $\alpha_n$ is the probability that the assignment output by the $\tau$-decimation algorithm \texttt{DEC\textsubscript{$\tau$}} on the random $k$-XORSAT instance $\mathbf{\Phi} \sim \mathbf{\Phi}_k(n,rn)$ with $n$ variables and $rn$ clauses is a solution for $\mathbf{\Phi}$.
Formally, we define $\alpha_n$ by the following expression
    $\alpha_n \equiv \proba{\sigma_{\mathbf{\Phi} \sim \mathbf{\Phi}_k(n,rn),\mathbf{Z},\mathbf{U}} \in \mathcal{S}(\mathbf{\Phi})}$,
where $\mathbf{\Phi} \sim \mathbf{\Phi}_k(n,rn)$ is the random $k$-XORSAT instance, $\mathbf{Z}$ is the random ordering vector, and $\mathbf{U}$ is the random internal vector, as mentioned in Section \ref{sec:sequence-of-output-assignments}.
Now, we consider the sequence of output assignments $\{\sigma_i=\sigma_{\mathbf{\Phi},\mathbf{Z},\mathbf{U^i}}\}_{i=0}^n$ generated by the procedure in Section \ref{sec:sequence-of-output-assignments}.
We first prove that if the algorithm \texttt{DEC\textsubscript{$\tau$}} is $\delta$-free, then the expected distance $\expected{d(\sigma_0,\sigma_n)}$ between the first and the last assignments in the sequence is at least $(\delta /2)n + o(n)$.

\begin{lemma}
    \label{lemma:delta-far}
    If the $\tau$-decimation algorithm \textnormal{\texttt{DEC\textsubscript{$\tau$}}} is $\delta$-free on the random $k$-XORSAT instance $\mathbf{\Phi} \sim \mathbf{\Phi}_k(n,rn)$ for some $\delta > 0$, then we have
    $\expected{ d( \sigma_0, \sigma_n ) } \ge (\delta/2)n + o(n)$.
\end{lemma}

Next, we will show that, if the $\tau$-decimation algorithm is "free enough", namely, strictly $2\mu(k,r)$-free, then we can pick a pair of output assignments and project them to the core instance $\mathbf{\Phi}_c$ so that the distance between the two corresponding core assignments falls in the \emph{forbidden range} from the overlap gap property of the core instance $\mathbf{\Phi}_c$.

\begin{lemma}
    \label{lemma:correct-distance}
    For any $k \ge 3$ and $r \in (r_{core}(k),r_{sat}(k))$, if the $\tau$-decimation algorithm \textnormal{\texttt{DEC\textsubscript{$\tau$}}} is strictly $2\mu(k,r)$-free on the random $k$-XORSAT instance $\mathbf{\Phi} \sim \mathbf{\Phi}_k(n,rn)$, then
    there exist $0 \le i_0 \le n$ and $0 < \epsilon' < \epsilon(k,r)$ such that w.h.p. we have
        $\left| d(\pi(\sigma_0), \pi(\sigma_{i_0})) - \frac{1}{2}\epsilon'n \right| < \frac{1}{4}\epsilon'n$,
    where $\epsilon(k,r)$ is given in Lemma \ref{lemma:ogp_of_core}.
\end{lemma}

The following lemma shows that the probability of both the output assignments $\sigma_{\mathbf{\Phi},\mathbf{Z},\mathbf{U^{0}}}$ and $\sigma_{\mathbf{\Phi},\mathbf{Z},\mathbf{U^{i_0}}}$ being solutions for the instance $\mathbf{\Phi}$ is lower bounded by $\alpha_n^2$.

\begin{lemma}
    \label{lemma:both-satisfying}
    For any $i \in [n]$, we have $\proba{\sigma_{0} \in \mathcal{S}(\Phi) \textnormal{\;and\;} \sigma_{i} \in \mathcal{S}(\Phi)} \ge \alpha_n^2$.
\end{lemma}

Finally, we can combine all above lemmas in this section to give the proof of Theorem \ref{thm:main-theorem}.

\begin{proof}[Proof of Theorem \ref{thm:main-theorem}]
    We denote by $\mathcal{A}$ the event of 
        $\left| d(\pi(\sigma_{0}), \pi(\sigma_{i_0})) - \frac{1}{2}\epsilon'n \right| < \frac{1}{4}\epsilon'n$,
    and we have $\proba{ \mathcal{A} }$ $= 1 - o(1)$ by Lemma \ref{lemma:correct-distance}.
    On the other hand, we pick $i=i_0$ for the inequality in Lemma \ref{lemma:both-satisfying}.
    We denote by $\mathcal{B}$ the event of
        $\sigma_{0} \in \mathcal{S}(\Phi) \textnormal{\;and\;} \sigma_{i_0} \in \mathcal{S}(\Phi)$,
    and $\proba{ \mathcal{B} } \ge \alpha_n^2$.
    Note that we have
    $\proba{\mathcal{A} \cap \mathcal{B}}
        \ge 1 - \proba{\text{Not\;}\mathcal{A}} - \proba{\text{Not\;}\mathcal{B}} 
        \ge 1 - o(1)  - ( 1 - \alpha_n^2 )
        = \alpha_n^2 - o(1)$.
    Thus, we have $\alpha_n \le \proba{ \mathcal{A} \cap \mathcal{B} }^{1/2} + o(1)$.

    Now assume both $\mathcal{A}$ and $\mathcal{B}$ take places.
    Since both $\sigma_{0}$ and $\sigma_{i_0}$ are solutions for the random instance $\mathbf{\Phi}$, both $\pi(\sigma_{0})$ and $\pi(\sigma_{i_0})$ are solutions for the core instance $\mathbf{\Phi}_c$.
    Moreover, the distance $d(\pi(\sigma_{0}), \pi(\sigma_{i_0}))$ falls in the interval $((1/4)\epsilon' n,(3/4)\epsilon' n) \subsetneq (o(n), \epsilon n)$, which takes place with probability at most $o(1)$ by Lemma \ref{lemma:ogp_of_core}.
    So, we have $\proba{\mathcal{A} \cap \mathcal{B}} \le o(1)$, and thus $\alpha_n \le o(1)$.
\end{proof}

To prove Theorem \ref{thm:main-theorem-unit-clause} and \ref{thm:main-theorem-exact-marginal}, all we need to do is to show that \texttt{DEC\textsubscript{UC}} and \texttt{DEC$_\tau$} with the exact marginal assumption are strictly $2\mu(k,r)$-free.
The results immediately follow by applying Theorem \ref{thm:main-theorem}.
From Lemma \ref{lemma:uc-free} and \ref{lemma:bp-free}, we know that \texttt{DEC\textsubscript{UC}} and \texttt{DEC$_\tau$} are $w_1(k,r)$-free and $w_e(k,r)$-free, respectively.
So, we only need to show that $w_1(k,r) > 2\mu(k,r)$ and $w_e(k,r) > 2\mu(k,r)$.
It can be done with the following lemmas, which give an upper bound of $\mu(k,r)$ in Lemma \ref{lemma:upper_bound_of_diameter}, a lower bound of $w_1(k,r)$ in Lemma \ref{lemma:lower-bound-w-1}, and a lower bound of and $w_e(k,r)$ in Lemma \ref{lemma:lower-bound-w-b}.
\ifappendix
The proofs of these three lemmas are given in Appendices \ref{sec:upper-bound-of-diameter}, \ref{sec:proof-of-lower-bound-w-1} and \ref{sec:proof-of-lower-bound-w-b}, respectively.
\fi

\begin{lemma}
    \label{lemma:upper_bound_of_diameter}
    For any $k \ge 4$ and $r \in (r_{core}(k),r_{sat}(k))$, we have $\mu(k,r) < \mu_u(k)$, where
    \begin{align*}
        \mu_u(k) = (1-e^{-1/k})-(1-e^{-1/k}) \ln (1-e^{-1/k}).
    \end{align*}
\end{lemma}

\begin{lemma}
    \label{lemma:lower-bound-w-1}
    For any $k \ge k_0 \ge 3$ and $r \in (r_{core}(k),r_{sat}(k))$, $w_1(k,r) \ge w_1^*(k_0)$, where
    \begin{align*}
        w_1^*(k) = \frac{k^{\frac{1}{1-k}}}{k-1} \gamma \left( \frac{1}{k-1}, k \left(\frac{k}{k+1}\right)^{k-1} \right)
    \end{align*}
\end{lemma}

\begin{lemma}
    \label{lemma:lower-bound-w-b}
    For any $k \ge k_0 \ge 3$ and $r \in (r_{core}(k),r_{sat}(k))$, we have $w_e(k,r) \ge w_e^*(k_0,r_{sat}(k_0))$, where $w_e^*(k,r) = x^-(k,r) - kr^2 (x^-(k,r))^{k}$ and
    \begin{align}
        x^\pm(k,r) = \left( \frac{1 \pm \sqrt{1-4(kr)^{-2}[(kr)^{\frac{1}{k-1}}-1]}}{2} \right)^{\frac{1}{k-2}}.
    \end{align}
\end{lemma}

\begin{proof}[Proof of Theorem \ref{thm:main-theorem-unit-clause}]
    Let $k \ge 9$ and $r \in (r_{core}(k),r_{sat}(k))$.
    By Lemma \ref{lemma:upper_bound_of_diameter} and \ref{lemma:lower-bound-w-1}, we have
    $2\mu(k,r) < 2\mu_u(9) \le 0.3420 < 0.3575 \le w_1^*(9) \le w_1(k,r)$.
    Then, by Lemma \ref{lemma:uc-free}, \texttt{DEC\textsubscript{UC}} is strictly $2\mu(k,r)$-free.
    The result follows.
\end{proof}

\begin{proof}[Proof of Theorem \ref{thm:main-theorem-exact-marginal}]
    Let $k \ge 13$ and $r \in (r_{core}(k),r_{sat}(k))$.
    By Lemma \ref{lemma:upper_bound_of_diameter} and \ref{lemma:lower-bound-w-b}, we have
    $2\mu(k,r) < 2\mu_u(13) \le 0.2668 < 0.2725 \le w_e^*(13) \le w_e(k,r)$.
    Then, by Lemma \ref{lemma:bp-free}, \texttt{DEC$_\tau$} is strictly $2\mu(k,r)$-free.
    The result follows.
\end{proof}

\bibliographystyle{alpha}
\bibliography{paper01}

\begin{appendices}
\section{Proof of lemmas in Section \ref{sec:proof-of-main-theorems}}

In this section, we prove the proof of lemmas in Section \ref{sec:proof-of-main-theorems}.

\begin{proof}[Proof of Lemma \ref{lemma:delta-far}]
    Suppose $D \subseteq [n]$ is the subset of indices of iterations that are free steps, namely,
    \begin{align*}
        D = \{i \in [n]: \text{The $i$-th iteration is a free step}\}.
    \end{align*}
    Note that 
    $\sigma_0 = \sigma_{\mathbf{\Phi},\mathbf{Z},\mathbf{U^0}} = \sigma_{\mathbf{\Phi},\mathbf{Z},\mathbf{V}}$ and 
    $\sigma_n = \sigma_{\mathbf{\Phi},\mathbf{Z},\mathbf{U^n}} = \sigma_{\mathbf{\Phi},\mathbf{Z},\mathbf{W}}$
    and thus
    \begin{align*}
        d( \sigma_0, \sigma_n )
        = d( \sigma_{\mathbf{\Phi},\mathbf{Z},\mathbf{U^0}}, \sigma_{\mathbf{\Phi},\mathbf{Z},\mathbf{U^n}} )
        = d( \sigma_{\mathbf{\Phi},\mathbf{Z},\mathbf{V}}, \sigma_{\mathbf{\Phi},\mathbf{Z},\mathbf{W}} )
    \end{align*}
    Since we can write $d( \sigma_{\mathbf{\Phi},\mathbf{Z},\mathbf{V}}, \sigma_{\mathbf{\Phi},\mathbf{Z},\mathbf{W}} ) = \sum_{i=1}^n \mathbbm{1}\left( \sigma_{\mathbf{\Phi},\mathbf{Z},\mathbf{V}}(x_{s(i)}) \neq \sigma_{\mathbf{\Phi},\mathbf{Z},\mathbf{W}}(x_{s(i)}) \right) $, we have
    \begin{align*}
        d( \sigma_{\mathbf{\Phi},\mathbf{Z},\mathbf{V}}, \sigma_{\mathbf{\Phi},\mathbf{Z},\mathbf{W}} )
        &= \sum_{i=1}^n \mathbbm{1}\left( \sigma_{\mathbf{\Phi},\mathbf{Z},\mathbf{V}}\left(x_{s(i)}\right) \neq \sigma_{\mathbf{\Phi},\mathbf{Z},\mathbf{W}}\left(x_{s(i)}\right) \right) \\
        &\ge \sum_{i \in D} \mathbbm{1}\left( \sigma_{\mathbf{\Phi},\mathbf{Z},\mathbf{V}}\left(x_{s(i)}\right) \neq \sigma_{\mathbf{\Phi},\mathbf{Z},\mathbf{W}}\left(x_{s(i)}\right) \right).
    \end{align*}
    In free steps, the local rule $\tau$ gives the value $1/2$ to the decimation algorithm.
    Therefore, for any $i \in D$, $\sigma_{\mathbf{\Phi},\mathbf{Z},\mathbf{V}}\left(x_{s(i)}\right) \neq \sigma_{\mathbf{\Phi},\mathbf{Z},\mathbf{W}}\left(x_{s(i)}\right)$ if and only if either
    $\mathbf{V}_i < 1/2 < \mathbf{W}_i$ or $\mathbf{W}_i < 1/2 < \mathbf{V}_i$.
    Therefore, we have 
    \begin{align*}
        \sum_{i \in D} \mathbbm{1}\left( \sigma_{\mathbf{\Phi},\mathbf{Z},\mathbf{V}}\left(x_{s(i)}\right) \neq \sigma_{\mathbf{\Phi},\mathbf{Z},\mathbf{V}}\left(x_{s(i)}\right) \right)
        &= \sum_{i \in D} \mathbbm{1}\left( \mathbf{V}_i < 1/2 < \mathbf{W}_i \text{\;or\;} \mathbf{W}_i < 1/2 < \mathbf{V}_i \right)
    \end{align*}
    Note that the random variables $\mathbf{V}_1, \mathbf{V}_2, \cdots, \mathbf{V}_n, \mathbf{W}_1, \mathbf{W}_2, \cdots, \mathbf{W}_n$ are i.i.d. over uniform distributions on $[0,1]$.
    Thus, $\sum_{i \in D} \mathbbm{1}\left( \mathbf{V}_i < 1/2 < \mathbf{W}_i \text{\;or\;} \mathbf{W}_i < 1/2 < \mathbf{V}_i \right)$ is distributed over the binomial distribution $B(|D|,1/2)$ with parameters $|D|$ and $1/2$.

    Assume that the algorithm \texttt{DEC\textsubscript{$\tau$}} is $\delta$-free on the random $k$-XORSAT instance $\mathbf{\Phi} \sim \mathbf{\Phi}_k(n,rn)$ for some $\delta > 0$, which implies that w.h.p. $|D| \ge \delta n$.
    Note that $D$ only depends on $\mathbf{\Phi}$ and $\mathbf{Z}$.
    Hence, we have
    \begin{align*}
        \expected{ d(\sigma_0, \sigma_n) }
        &= \mathbb{E}_{\mathbf{\Phi},\mathbf{Z},\mathbf{V},\mathbf{W}} \left[ \sum_{i \in D} \mathbbm{1}\left( \mathbf{V}_i < 1/2 < \mathbf{W}_i \text{\;or\;} \mathbf{W}_i < 1/2 < \mathbf{V}_i \right) \right] \\
        &= \mathbb{E}_{\mathbf{V},\mathbf{W}} \mathbb{E}_{\mathbf{\Phi},\mathbf{Z}} \left[ \sum_{i \in D} \mathbbm{1}\left( \mathbf{V}_i < 1/2 < \mathbf{W}_i \text{\;or\;} \mathbf{W}_i < 1/2 < \mathbf{V}_i \right) \right] \\
        &= (1/2) \cdot \expected{|D|} \\
        &\ge (\delta/2)n + o(n)
    \end{align*}
\end{proof}

\begin{proof}[Proof of Lemma \ref{lemma:correct-distance}]
    Assume the $\tau$-decimation algorithm \texttt{DEC$_\tau$} is strictly $2\mu(k,r)$-free on the random $k$-XORSAT instance $\mathbf{\Phi} \sim \mathbf{\Phi}_k(n,rn)$.
    Then, there exists $\delta > 2\mu(k,r)$ such that \texttt{DEC$_\tau$} is $\delta$-free on $\mathbf{\Phi}$.
    
    We first show that there exists $0 \le i_0 \le n$ such that the expected value of $d(\pi(\sigma_{0}), \pi(\sigma_{i_0}))$ is close to $\frac{1}{2} \epsilon' n$ for some $\epsilon' < \epsilon(k,r)$.
    Consider the sequence $\{\expected{d(\pi(\sigma_{0}), \pi(\sigma_{i}))}\}_{i=0}^n$ in which the first item is
    \begin{align}
        \expected{ d(\pi(\sigma_{0}), \pi(\sigma_{0})) } = 0.
        \label{eqt:expectation_first}
    \end{align}
    By Lemma \ref{lemma:delta-far}, we know that $\expected{ d(\sigma_{0}, \sigma_{n}) } \ge (\delta/2)n + o(n)$.
    Moreover, by Theorem \ref{thm:core_emerges}, we know that w.h.p. there are $\mu(k,r)n + o(n)$ variables not in the core instance $\mathbf{\Phi}_c$.
    Note that the projection function $\pi$ only remove variables not in the core instance $\mathbf{\Phi}_c$.
    Therefore, we have $d( \pi(\sigma_0), \pi(\sigma_n) ) \ge d( \sigma_0, \sigma_n ) - n^*$ where $n^*$ is the number of variables not in the core instance $\mathbf{\Phi}_c$.
    Therefore, we have
    \begin{align*}
        \expected{d( \sigma_0, \sigma_n )}
        \le \expected{d( \pi(\sigma_0), \pi(\sigma_n) )} + \expected{n^*}
        \le \expected{d( \pi(\sigma_0), \pi(\sigma_n) )} + \mu(k,r)n + o(n)
    \end{align*}
    By re-assigning the terms, we have
    \begin{align}
        \expected{d( \pi(\sigma_0), \pi(\sigma_n) )}
        &\ge (\delta/2 - \mu(k,r)) n + o(n)
        \label{eqt:expectation_last}
    \end{align}
    with $\delta/2 - \mu(k,r) > 0$.
    From Lemma \ref{lemma:insensitive}, we know that with probability $1 - \exp(-\ln n (\ln \ln n)^{\xi/4})$ we have
    \begin{align}
        d(\sigma_{i-1}, \sigma_{i}) < n^{1/6} \text{\quad for all\;}i \in [n].
        \label{eqt:influence-range-size-in-lemma}
    \end{align}
    By the triangle inequality of the metric $d$ and the linearity of the expectation, we know that for $1 \le i \le n$ the difference of two consecutive expected values in the sequence is 
    \begin{align}
        &\expected{ d(\pi(\sigma_{0}), \pi(\sigma_{i})) } - \expected{ d(\pi(\sigma_{0}), \pi(\sigma_{i-1})) } \nonumber \\
        &\le \expected{ d(\pi(\sigma_{0}), \pi(\sigma_{i-1})) } + \expected{ d(\pi(\sigma_{i-1}), \pi(\sigma_{i})) } - \expected{ d(\pi(\sigma_{0}), \pi(\sigma_{i-1})) } \nonumber \\
        &= \expected{ d(\pi(\sigma_{i-1}), \pi(\sigma_{i})) } \nonumber \\
        &\le \expected{ d(\sigma_{i-1}, \sigma_{i}) } \nonumber \nonumber \\
        &\le n^{1/6} + o(1). \label{eqt:expectation_diff}
    \end{align}
    Combining (\ref{eqt:expectation_first}), (\ref{eqt:expectation_last}) and (\ref{eqt:expectation_diff}), we know that there exists $0 \le i_0 \le n$ and $0 < \epsilon' < \min\{\delta/2-\mu(k,r),\epsilon(k,r)\}$ such that 
    \begin{align}
        \expected{ d(\pi(\sigma_{0}), \pi(\sigma_{i_0})) }
        \in \left[ \frac{1}{2}\epsilon'n,\, \frac{1}{2}\epsilon'n + n^{1/6} \right].
    \end{align}

    Next, we prove that $d(\pi(\sigma_{0}), \pi(\sigma_{i_0}))$ concentrates around its mean.
    Given two vectors $A \in \{0,1\}^n$ and $B \in \{0,1\}^n$, we write $A \cdot B = (A_1,\cdots,A_n,B_1,\cdots,B_n)$ and $A \oplus_i B = (A_1,\cdots,A_{i},B_{i+1},\cdots,B_n)$ for $i \in [n]$.
    A function $f:D_1 \times D_2 \times \cdots D_n \rightarrow \mathbb{R}$ satisfies the \textbf{bounded differences} property if there exist $c_1,c_2,\cdots,c_n \in \mathbb{R}$ such that for any $x_1 \in D_1$, $x_2 \in D_2$, $\cdots$, $x_n \in D_n$, $y_i \in D_i$ and $i \le [n]$,
    \begin{align*}
        | f(x_1,...,x_{i-1},x_i,x_{i+1},...,x_n) - f(x_1,...,x_{i-1},y_i,x_{i+1},...,x_n) | \le c_i.
    \end{align*}
    Note that $\mathbf{U^{i_0}} = \mathbf{W} \oplus_{i_0} \mathbf{V}$.
    For arbitrary instance $\Phi$ and ordering vector $Z$, we have
    \begin{align*}
        d(\pi(\sigma_{\Phi,Z,\mathbf{U^0}}), \pi(\sigma_{\Phi,Z,\mathbf{U^{i_0}}}))
        &= d(\pi(\sigma_{\Phi,Z,\mathbf{V}}), \pi(\sigma_{\Phi,Z,\mathbf{W} \oplus_{i_0} \mathbf{V}})).
    \end{align*}
    So, given a random instance $\mathbf{\Phi}$ and a random ordering vector $\mathbf{Z}$, we can write $d(\pi(\sigma_{0}), \pi(\sigma_{i_0}))$ as a function $f:\{0,1\}^{2n} \rightarrow \mathbb{R}$ on variables $\mathbf{V}\cdot\mathbf{W}$ given by
    $f(\mathbf{V} \cdot \mathbf{W}) = d(\pi(\sigma_{\mathbf{\Phi},\mathbf{Z},\mathbf{V}}), \pi(\sigma_{\mathbf{\Phi},\mathbf{Z},\mathbf{W} \oplus_{i_0} \mathbf{V}}))$.
    Conditioning on (\ref{eqt:influence-range-size-in-lemma}), we can verify that $f$ satisfies bounded differences property with $c_i=2n^{1/6}$ for $i \in [n]$, and thus we have
    \begin{align*}
        &\proba{ \left| d(\pi(\sigma_{0}), \pi(\sigma_{i_0})) - \frac{1}{2}\epsilon'n \right| \ge \frac{1}{4}\epsilon'n } \\
        &\le \proba{ \left| d(\pi(\sigma_{0}), \pi(\sigma_{i_0})) - \expected{ d(\pi(\sigma_{0}), \pi(\sigma_{i_0})) } \right| \ge \frac{1}{4}\epsilon'n - n^{1/6} } \\
        &= \proba{ \left| f(V \cdot W) - \expected{ f(V \cdot W) } \right| \ge \frac{1}{4}\epsilon'n - n^{1/6} } \\
        &\le 2 \exp \left( - \frac{2 \left( \frac{1}{4}\epsilon'n - n^{1/6} \right)^2}{(n+i_0)n^{1/6}}  \right) \\
        &\le 2 \exp \left( - \frac{1}{8} n^{5/6} + o(n^{5/6}) \right)
    \end{align*}
    by McDiarmid's inequality.
    Since the condition (\ref{eqt:influence-range-size-in-lemma}) holds with probability $1-\exp(-\ln n(\ln \ln n)^{\xi/4}) \rightarrow 1$ as $n \rightarrow \infty$, the inequality in Lemma \ref{lemma:correct-distance} holds with high probability.
\end{proof}

\begin{proof}[Proof of Lemma \ref{lemma:both-satisfying}]
    Fix an arbitrary $i \in [n]$.
    Note that we have $\mathbf{U^{0}} = (\mathbf{V}_1,\mathbf{V}_2,\cdots,\mathbf{V}_n)$ and $\mathbf{U^{i}} = (\mathbf{W}_1,\mathbf{W}_2,\cdots,\mathbf{W}_i,$ $\mathbf{V}_{i+1},\cdots,\mathbf{V}_n)$, where $\mathbf{V}_1,\mathbf{V}_2,\cdots,\mathbf{V}_n,\mathbf{W}_1,\mathbf{W}_2,\cdots,\mathbf{W}_n$ are uniformly distributed over $[0,1]$, independently.
    Conditioning on $\mathbf{\Phi},\mathbf{Z},\mathbf{V}_{i+1},\cdots,\mathbf{V}_n$, the assignment $\sigma_{\mathbf{\Phi},\mathbf{Z},\mathbf{U^{0}}}$ only depends on $\mathbf{V}_1,\cdots,\mathbf{V}_i$, and the assignment $\sigma_{\mathbf{\Phi},\mathbf{Z},\mathbf{U^{i}}}$ only depends on $\mathbf{W}_1,\cdots,\mathbf{W}_i$, and we have
    \begin{align*}
        &\mathbb{E}_{\mathbf{V}_1,\cdots,\mathbf{V}_i, \mathbf{W}_1,\cdots,\mathbf{W}_i}\left[\; \mathbbm{1}( \sigma_{\mathbf{\Phi},\mathbf{Z},\mathbf{U^{0}}} \in \mathcal{S}(\Phi) ) \cdot \mathbbm{1} ( \sigma_{\mathbf{\Phi},\mathbf{Z},\mathbf{U^{i}}} \in \mathcal{S}(\Phi) ) \;\right] \\
        &= \left( \mathbb{E}_{\mathbf{V}_1,\cdots,\mathbf{V}_i} \mathbbm{1}( \sigma_{\mathbf{\Phi},\mathbf{Z},\mathbf{U^{0}}} \in \mathcal{S}(\Phi) ) \right) \cdot \left( \mathbb{E}_{\mathbf{W}_1,\cdots,\mathbf{W}_i} \mathbbm{1} ( \sigma_{\mathbf{\Phi},\mathbf{Z},\mathbf{U^{i}}} \in \mathcal{S}(\Phi) ) \right) \\
        &= \left( \mathbb{E}_{\mathbf{V}_1,\cdots,\mathbf{V}_i} \mathbbm{1}( \sigma_{\mathbf{\Phi},\mathbf{Z},\mathbf{V}} \in \mathcal{S}(\Phi) ) \right)^2
    \end{align*}
    Therefore, by Jensen's inequality, we have
    \begin{align*}
        &\Pr_{\mathbf{\Phi},\mathbf{Z},\mathbf{V},\mathbf{W}}\left[\; \sigma_{0} \in \mathcal{S}(\Phi) \textnormal{\;and\;} \sigma_{i} \in \mathcal{S}(\Phi) \;\right] \\
        &=\Pr_{\mathbf{\Phi},\mathbf{Z},\mathbf{V},\mathbf{W}}\left[\; \sigma_{\mathbf{\Phi},\mathbf{Z},\mathbf{U^{0}}} \in \mathcal{S}(\Phi) \textnormal{\;and\;} \sigma_{\mathbf{\Phi},\mathbf{Z},\mathbf{U^{i}}} \in \mathcal{S}(\Phi) \;\right] \\
        &= \mathbb{E}_{\mathbf{\Phi},\mathbf{Z},\mathbf{V}_{i+1},\cdots,\mathbf{V}_n} \mathbb{E}_{\mathbf{V}_1,\cdots,\mathbf{V}_i, \mathbf{W}_1,\cdots,\mathbf{W}_i}\left[\; \mathbbm{1}( \sigma_{\mathbf{\Phi},\mathbf{Z},\mathbf{U^{0}}} \in \mathcal{S}(\Phi) ) \cdot \mathbbm{1} ( \sigma_{\mathbf{\Phi},\mathbf{Z},\mathbf{U^{i}}} \in \mathcal{S}(\Phi) ) \;\right] \\
        &= \mathbb{E}_{\mathbf{\Phi},\mathbf{Z},\mathbf{V}_{i+1},\cdots,\mathbf{V}_n} \left( \mathbb{E}_{\mathbf{V}_1,\cdots,\mathbf{V}_i} \mathbbm{1}( \sigma_{\mathbf{\Phi},\mathbf{Z},\mathbf{V}} \in \mathcal{S}(\Phi) ) \right)^2 \\
        &\ge \left( \mathbb{E}_{\mathbf{\Phi},\mathbf{Z},\mathbf{V}_{i+1},\cdots,\mathbf{V}_n} \mathbb{E}_{\mathbf{V}_1,\cdots,\mathbf{V}_i} \mathbbm{1}( \sigma_{\mathbf{\Phi},\mathbf{Z},\mathbf{V}} \in \mathcal{S}(\Phi) ) \right)^2 \\
        &= \left( \Pr_{\mathbf{\Phi},\mathbf{Z},\mathbf{V}}\left[\; \sigma_{\mathbf{\Phi},\mathbf{Z},\mathbf{V}} \in \mathcal{S}(\Phi) \;\right] \right)^2 \\
        &= \alpha_n^2.
    \end{align*}
\end{proof}

\section{Upper Bound of Cluster Diameter}
\label{sec:upper-bound-of-diameter}

In this section, we give the proof of Lemma \ref{lemma:upper_bound_of_diameter}, which gives an upper bound for the diameter of a cluster that can be written as a closed form expression.
From Lemma \ref{lemma:partition-sol-space}, we know that for any $k \ge 3$ and $r_{core}(k) < r < r_{sat}(k)$ w.h.p. the diameter of a cluster is upper bounded by $\mu(k,r) n + o(n)$, where
\begin{align*}
    \mu(k,r) = \exp(-krQ_{k,r}^{k-1}) + krQ_{k,r}^{k-1} \exp(-krQ_{k,r}^{k-1})
\end{align*}
and $Q_{k,r}$ is the largest solution of the fixed point equation $Q = 1 - \exp(-krQ^{k-1})$ with the given values of $k$ and $r$.
This implicit expression would make calculations complicated.
So, we slightly relax the upper bound to obtain a simpler expression $\mu_u(k)$ in Lemma \ref{lemma:upper_bound_of_diameter}.
Note that this lemma can only be applied when $k \ge 4$ due to some calculation restriction, which will be mentioned later in this section.

To prove Lemma \ref{lemma:upper_bound_of_diameter}, we first re-write the fixed point equation $Q=1-\exp(-krQ^{k-1})$ as 
\begin{align*}
    krQ^{k-1} = - \ln(1-Q).
\end{align*}
Since $Q_{k,r}$ satisfies this equation, we can write $\mu(k,r)$ as
\begin{align}
    \mu(k,r) \nonumber
    &= \exp(-krQ_{k,r}^{k-1}) + krQ_{k,r}^{k-1} \exp(-krQ_{k,r}^{k-1}) \nonumber \\
    &= (1-Q_{k,r}) - (1-Q_{k,r}) \ln (1-Q_{k,r}) 
\end{align}
Note that $Q_{k,r}$ must lie in the interval $[0,1)$ as
\begin{align*}
    Q \le 0 < 1-\exp(-krQ^{k-1}) &\text{\quad for any\;} Q \le 0 \quad{\;and} \\
    Q > 1 \ge 1-\exp(-krQ^{k-1}) &\text{\quad for any\;} Q > 1.
\end{align*}
Further note that the real-valued function $f(x)=(1-x)-(1-x)\ln(1-x)$ is strictly decreasing on $[0,1)$.
So, it suffices to find a lower bound of $Q_{k,r}$, in order to find an upper bound of $\mu(k,r)$.

Next, we try to show that $e^{-1/k}$ is a lower bound of $Q_{k,r}$.
To facilitate the calculation, we define a real-valued analytic function $G:[3,+\infty) \times [0,1] \times [0,1] \rightarrow \mathbb{R}$ by
\begin{align*}
    G(k,r,Q) = 1 - \exp(-krQ^{k-1}) -Q,
\end{align*}
which have the following derivatives:
\begin{align}
    \label{eq:dG_dr}
    \frac{\partial G}{\partial r}
    &= \exp(-krQ^{k-1}) \cdot kQ^{k-1} \\
    \frac{\partial G}{\partial Q}
    &= \exp(-krQ^{k-1}) \cdot kr(k-1) Q^{k-2} - 1 \\
    \label{eq:d2G_dQ2}
    \frac{\partial^2 G}{\partial Q^2}
    &= \exp(-krQ^{k-1}) \cdot kr(k-1)Q^{k-3} [(k-2)-kr(k-1)Q^{k-1}]
\end{align}
Before proving that $e^{-1/k}$ is a lower bound of $Q_{k,r}$, we give the following two lemmas, Lemma \ref{lemma:G_k_rc_e1k_le_0} and Lemma \ref{lemma:existence_r1}, which will be used later.

\begin{lemma}
    \label{lemma:G_k_rc_e1k_le_0}
    For any $k \ge 3$, $G(k,r_{core}(k),e^{-1/k}) \le 0$.
\end{lemma}

\begin{proof}
    We prove it by contradiction.
    Assume $G(k,r_{core}(k),e^{-1/k}) > 0$ for some $k \ge 3$.
    By the continuity of $G$, there exists $r' < r_{core}(k)$ such that $G(k,r',e^{-1/k}) > 0$.
    Note that $G(k,r',1) = 1-\exp(-kr')-1 < 0$.
    Again, by the continuity of $G$, there exists $Q' \in (e^{-1/k},1)$ such that $G(k,r',Q')=0$.
    Since $0 < r' < r_{core}(k)$, this contradicts the definition of $r_{core}(k)$.
\end{proof}

\begin{lemma}
    \label{lemma:existence_r1}
    For any $k \ge 3$, there exists $r_1(k) \in (r_{core}(k),1)$ such that
    \begin{align*}
        G(k,r,e^{-1/k})
        \begin{cases}
            <0 & \text{\quad for \;} r_{core}(k) \le r < r_1(k) \\
            =0 & \text{\quad for \;} r = r_1(k) \\
            >0 & \text{\quad for \;} r_1(k) < r \le 1
        \end{cases}  
    \end{align*}
\end{lemma}

\begin{proof}
    Note that $\lim_{r \rightarrow \infty} G(k,r,e^{-1/k}) = \lim_{r \rightarrow \infty} 1-\exp(-kre^{-(k-1)/k})-e^{-1/k} = 1-e^{-1/k} > 0$.
    That means for sufficiently large $r$, $G(k,r,e^{-1/k}) > 0$.
    On the other hands, from Lemma \ref{lemma:G_k_rc_e1k_le_0}, we know that $G(k,r_{core}(k),e^{-1/k}) \le 0$ for any $k \ge 3$.
    Moreover, from (\ref{eq:dG_dr}), we have $\frac{\partial}{\partial r} G(k,r,e^{-k}) > 0$, which implies $G(k,r,e^{-1/k})$ is strictly increasing with $r$ for $r>0$.
    Therefore, there exists $r_1 = r_1(k) \in (r_{core}(k),1)$ such that $G(k,r,e^{-1/k}) < 0$ for $r_{core}(k) \le r < r_1(k)$, $G(k,r,e^{-1/k}) > 0$ for $r_1(k) < r \le 1$ and $G(k,r_1(k),e^{-1/k}) = 0$.
\end{proof}

With Lemma \ref{lemma:G_k_rc_e1k_le_0} and Lemma \ref{lemma:existence_r1}, we can prove the $e^{-1/k}$ is a lower bound for $Q_{k,r}$.
We split the proof into two cases: the case of $r \in [r_{core}(k), r_1(k)]$, and the case of $r \in (r_1(k), 1]$.
We first study the latter case, which is easier to be proved.

\begin{lemma}
    \label{lemma:lower_bound_q_kr_case_2}
    For any $k \ge 3$ and $r \in (r_1(k),1]$, we have $Q_{k,r} > e^{-1/k}$.
\end{lemma}

\begin{proof}
    From Lemma \ref{lemma:existence_r1}, we have $G(k,r,e^{-1/k}) > 0$ since $r > r_1(k)$.
    Note that $G(k,r,1) = 1 - \exp(-kr) - 1 < 0$.
    By the continuity of $G$, there exists at least one $Q' \in (e^{-1/k}, 1)$ such that $G(k,r,Q')=0$, which can be written as $Q'=1-\exp(-kr(Q')^{k-1})$.
    By the definition of $Q_{k,r}$, we then have $Q_{k,r} \ge Q' > e^{-1/k}$.
\end{proof}

Next, we study the case of $r \in [r_{core}(k), r_1(k)]$.
To prove that $Q_{k,r} > e^{-1/k}$, we need the two following facts from the basic calculus.
With these two facts, we can prove that $Q_{k,r} > e^{-1/k}$ for $r \in [r_{core}(k), r_1(k)]$ in Lemma \ref{lemma:lower_bound_q_kr_case_1}.
Note that the condition $k \ge 4$ is required in the calculation in the proof of Lemma \ref{lemma:lower_bound_q_kr_case_1}.
This is the reason why Lemma \ref{lemma:upper_bound_of_diameter} requires $k \ge 4$.

\begin{fact}
    \label{fact:calculus1}
    Let $f:[a,b]\rightarrow \mathbb{R}$ be an analytic function.
    If $f''(x) > 0$ for any $x \in (a,b)$, then $\max_{x \in [a,b]} f(x) = \max \{ f(a), f(b) \}$.
\end{fact}

\begin{fact}
    \label{fact:calculus2}
    Let $f:[a,b]\rightarrow \mathbb{R}$ be an analytic function.
    If there exists $c \in (a,b)$ such that
    \begin{align*}
        f''(x)
        \begin{cases}
            >0 & \text{\quad for \;} x\in(a,c) \\
            =0 & \text{\quad for \;} x=c \\
            <0 & \text{\quad for \;} x\in(c,b)
        \end{cases}
    \end{align*}
    and $f'(b) \ge 0$, then $\max_{x \in [a,b]} f(x) = \max \{ f(a), f(b) \}$.
\end{fact}

\begin{lemma}
    \label{lemma:lower_bound_q_kr_case_1}
    For any $k \ge 4$ and $r \in [r_c,r_1]$, where $r_c=r_{core}(k)$ and $r_1=r_1(k)$, we have $Q_{k,r} > e^{-1/k}$.
\end{lemma}
\begin{proof}
    Given arbitrary $k \ge 3$ and $Q>0$, from (\ref{eq:dG_dr}), we know that $G(k,r,Q)$ is strictly increasing with $r \in [r_c,r_1]$, and thus we have $Q(k,r,Q) < Q(k,r_1,Q)$ for any $r \in [r_c,r_1]$.

    From (\ref{eq:d2G_dQ2}), when we view $G$ as a function of $Q$ with fixed $k \ge 3$ and $r>0$ there are two points of inflection:
    \begin{align*}
        Q=0 \text{\quad and \quad} Q=\left(\frac{k-2}{kr(k-1)}\right)^{1/(k-1)}
    \end{align*}
    Now, we consider the following two cases separately:
    \begin{enumerate}
        \item $\left(\frac{k-2}{kr(k-1)}\right)^{1/(k-1)} < e^{-1/k}$
        \item $e^{-1/k} \le \left(\frac{k-2}{kr(k-1)}\right)^{1/(k-1)}$
    \end{enumerate}

    \noindent \textbf{Case 1.} Assume $(\frac{k-2}{kr(k-1)})^{1/(k-1)} < e^{-1/k}$.
    Then we have
    \begin{align*}
        \frac{\partial^2}{\partial Q^2} G(k,r_1,Q)
        \begin{cases}
            >0 & \text{\quad for \;} 0 < Q < \left(\frac{k-2}{kr(k-1)}\right)^{1/(k-1)} \\
            =0 & \text{\quad for \;} Q = \left(\frac{k-2}{kr(k-1)}\right)^{1/(k-1)} \\
            <0 & \text{\quad for \;} \left(\frac{k-2}{kr(k-1)}\right)^{1/(k-1)} < Q < e^{-1/k}
        \end{cases}
    \end{align*}
    Moreover, since $G(k,r_1,e^{-1/k})=0$, we have $e^{-1/k} = 1 - \exp(-kr_1e^{-(k-1)/k})$.
    Therefore, we have
    \begin{align*}
        \left. \frac{\partial}{\partial Q} G(k,r_1,Q) \right|_{Q=e^{-1/k}}
        &= \exp(-kr_1e^{-(k-1)/k}) \cdot kr_1(k-1)e^{-(k-2)/k} - 1 \\
        &= \exp(-kr_1e^{-(k-1)/k}) \cdot kr_1e^{-(k-1)/k} \cdot (k-1)e^{1/k} - 1 \\
        &= (1-e^{-1/k}) \cdot [-\ln (1-e^{-1/k})] \cdot (k-1) e^{1/k} - 1 \\
        &> 0
    \end{align*}
    for $k \ge 4$.
    By Fact \ref{fact:calculus2}, we have $G(k,r_1,Q) < \max \{ G(k,r_1,0), G(k,r_1,e^{-1/k}) \}$ for any $G \in (0,e^{-1/k})$.
    Note that both $G(k,r_1,0)$ and $G(k,r_1,e^{-1/k})$ are less than 0.
    So, $G(k,r,Q) \le G(k,r_1,Q) < 0$ for any $k \ge 4$, $r \in [r_c,r_1]$ and $Q \in [0,e^{-1/k}]$.
    Therefore, $Q_{k,r}$ cannot be in $[0,e^{-1/k}]$ by its definition, and hence $Q_{k,r} > e^{-1/k}$.

    \bigskip

    \noindent \textbf{Case 2.} Assume $e^{-1/k} \le (\frac{k-2}{kr(k-1)})^{1/(k-1)}$.
    From (\ref{eq:d2G_dQ2}), we have
    \begin{align*}
        \frac{\partial^2}{\partial Q^2} G(k,r_1,Q) > 0
    \end{align*}
    for $Q \in [0,e^{-1/k}]$.
    By Fact \ref{fact:calculus1}, we know that for $G \in [0,e^{-1/k}]$
    \begin{align*}
        G(k,r_1,Q) \le \max \{ G(k,r_1,0), G(k,r_1,e^{-1/k}) \}.
    \end{align*}
    Since both $G(k,r_1,0)$ and $G(k,r_1,e^{-1/k})$ are less than 0, we have $G(k,r,Q) \le G(k,r_1,Q) < 0$ for any $k \ge 3$, $r \in [r_c,r_1]$ and $Q \in [0,e^{-1/k}]$.
    Therefore, $Q_{k,r}$ cannot be in $[0,e^{-1/k}]$ by its definition, and hence $Q_{k,r} > e^{-1/k}$.
\end{proof}

Now, we can complete the proof for Lemma \ref{lemma:upper_bound_of_diameter}.

\begin{proof}[Proof of Lemma \ref{lemma:upper_bound_of_diameter}.]
    Let $k \ge 4$ and $r \in (r_{core}(k),r_{sat}(k))$.
    From the fixed point equation $Q = 1-\exp(-krQ^{k-1})$, we have $krQ_{k,r}^{k-1}=-\ln(1-Q_{k,r})$.
    Note that the real-valued function $(1-x)-(1-x)\ln(1-x)$ is strictly decreasing on $[0,1)$, and $Q_{k,r} \in [0,1)$.
    By Lemma \ref{lemma:lower_bound_q_kr_case_2} and Lemma \ref{lemma:lower_bound_q_kr_case_1}, we have $Q_{k,r} > e^{-1/k}$.
    Combining all these, the result is followed by
    \begin{align*}
        \mu(k,r)
        &= \exp(-krQ_{k,r}^{k-1}) + krQ_{k,r}^{k-1} \exp(-krQ_{k,r}^{k-1}) \\
        &= (1-Q_{k,r}) - (1-Q_{k,r}) \ln (1-Q_{k,r}) \\
        &< (1-e^{-1/k}) - (1-e^{-1/k}) \ln (1-e^{-1/k}) \\
        &= \mu_u(k).
    \end{align*}
\end{proof}

\section{Terminologies}

In this section, we cover the definition of some terminologies used in Appendices \ref{sec:uc-free-steps} and \ref{sec:bp-free-steps}.

\subsection{Degree profile}
\label{sec:term-degree-profile}

The distribution of the degrees of nodes in a factor graph can be described by \emph{degree profile}, which plays an important role in our analysis.
Given a factor graph $G$, let $n_i$ be the number of variable nodes of degree $i$ for all $i$.
Similarly, let $m_i$ be the number of equations nodes of degree $i$ for all $i$.
Furthermore, let $\widehat{n}$ be the total number of variable nodes and $\widehat{m}$ be the total number of equation nodes.
It is obvious that $\widehat{n}=n$ and $\widehat{m}=m$ for the factor graph of a $k$-XORSAT instance, but keep in mind that the total number $\widehat{n}$ of variable nodes and the total number $\widehat{m}$ of equation nodes decrease during the process of the algorithm.
The \textbf{degree distribution of variable nodes} is given by the sequence $\Lambda=\{\Lambda_i\}_{i\ge 0}$, where $\Lambda_i = n_i/\widehat{n}$, and the \textbf{degree distribution of equation nodes} is given by the sequence $P=\{P_i\}_{i\ge 0}$, where $P_i = m_i/\widehat{m}$.
Then, the \textbf{degree profile} of the factor graph $G$ is given by $(\Lambda,P)$.
Sometimes, the degree distributions $\Lambda$ and $P$ can also be represented by the polynomials $\Lambda(x)=\sum_{i\ge 0}\Lambda_i x^i$ and $P(x)=\sum_{i\ge 0}P_i x^i$, respectively.
With this representation, we can write $\sum_{i \ge 1} i\Lambda_i = \Lambda'(1)$ and $\sum_{i \ge 1}i P_i = P'(1)$.

Given the factor graph $G$ of a random $k$-XORSAT instance $\mathbf{\Phi} \sim \mathbf{\Phi}_k(n,rn)$, $G$ is uniformly distributed over the ensemble of $k$-uniform factor graph $\mathbb{G}_n(k,rn)$.
It is clear that the degree distribution of equation nodes is given by $P(x)=x^k$.
The degree distribution of variable nodes is also known to converge in distribution to independent Poisson random variables with mean $kr$.
Precisely speaking, w.h.p. for any $0 \le i \le rn$, we have
\begin{align*}
    \Lambda_i = e^{-kr}\frac{(kr)^i}{i!} + o(1).
\end{align*}

We can also describe the degree profile from \emph{edge perspective}.
The \textbf{edge perspective degree distribution of variable nodes} is given by $\lambda=\{\lambda_i\}_{i\ge 1}$, where $\lambda_i=i\Lambda_i/\sum_{j\ge 1}j\Lambda_j$, and the \textbf{edge perspective degree distribution of equation nodes} is given by $\rho=\{\rho_i\}_{i\ge 1}$, where $\rho_i=iP_i/\sum_{j\ge 1}jP_j$.
Then, the \textbf{edge perspective degree profile} of the factor graph $G$ is given by $(\lambda,\rho)$.
Similar to $\Lambda$ and $P$, the edge perspective degree distribution $\lambda$ and $\rho$ can be written as the polynomials  $\lambda(x)=\sum_{i\ge 1}\lambda_i x^{i-1}$ and $\rho(x)=\sum_{i\ge 1}\rho_i x^{i-1}$, respectively.

The ensemble of factor graphs with prescribed degree profile is called the \textbf{ensemble of degree constrained factor graphs} $\mathbb{D}_n(\Lambda,P)$, which is the set of all factor graphs of $n$ variable nodes with degree profile $(\Lambda,P)$ with the uniform distribution.
Note that the number $m$ of the function nodes is restricted to satisfy the equation $\Lambda'(1)n=P'(1)m$.

\subsection{Local tree-like structure}

In a factor graph $G$, we can define the \textbf{length} of a path $(v_1,v_2,...,v_l)$ to be the number of edges in the path.
Then, the \textbf{distance} between two nodes is defined to be the length of the shortest path between them.
By convention, we set the length to be $+\infty$ between two nodes if there is no path connecting them.
With this notion of \emph{distance}, we can define the \emph{local neighborhood} of a node.
Given a factor graph $G$, the \textbf{local neighborhood} (or simply \textbf{neighborhood}) $B_G(x,R)$ of a node $x$ of radius $R \ge 0$ is defined to be the subgraph of $G$ induced by all nodes of distances at most $R$ from $x$ and all edges between those nodes.
The local neighborhood $B_G(x,R)$ also represents an XORSAT instance with the variables and clauses inside the neighborhood.

The local neighborhood of a variable in a random $k$-uniform factor graph looks like a tree.
In particular, it looks similar to a \emph{random $R$-generation tree}.
For any non-negative even number $R \ge 0$, the \textbf{$R$-generation tree ensemble} $\mathbb{T}_R(\Lambda,P)$ of a given degree profile $(\Lambda,P)$ is defined as follows.
When $R=0$, the ensemble contains only one element, a single isolated node, and call it the \textbf{variable node of the generation} 0.
Assume $R>0$.
We first generate a tree $T$ from the $(R-2)$-generation tree ensemble $\mathbb{T}_{R-2}(\Lambda,P)$.
For each variable node $x$ of generation $R-2$, we draw an independent integer $i \ge 1$ distributed according to $\lambda_i$ (or $\Lambda_i$ if $R=2$), and add $i-1$ function nodes, which are connected to $x$ as its children.
Then, for each of these function nodes $a$, we draw an independent integer $j \ge 1$ distributed according to $\rho_j$, and add $j-1$ variable nodes, which are connected to $a$ as its children and called the \textbf{variable nodes of the generation $R$}.

In particular, Mézard and Montanari \cite{mezardInformationPhysicsComputation2009} shows that the local structure of a variable node of a random factor graph from $\mathbb{D}_n(\Lambda,P)$ and $\mathbb{G}_n(k,m)$ converges to this tree ensemble.
The more details of the following theorem can be found in \cite{mezardInformationPhysicsComputation2009}.

\begin{theorem}
    \label{thm:tree_like_neighborhood}
    Let $(\Lambda,P)$ be a fixed degree profile, $G$ be a random factor graph in the $\mathbb{D}_n(\Lambda,P)$ ensemble (and $\mathbb{G}_n(k,m)$ respectively), $x$ be a variable node chosen uniformly at random from $G$, and $R$ be a non-negative even number.
    Then, the local neighborhood $B_G(x,R)$ of the factor graph $G$ converges in distribution to $\mathbb{T}_{R}(\Lambda,P)$ (and $\mathbb{T}_R(e^{kr(x-1)},x^k)$ respectively) as $n \rightarrow \infty$.
\end{theorem}

\section{Proof of Lemma \ref{lemma:uc-free}: Finding the number of free steps in \texttt{DEC\textsubscript{UC}}}
\label{sec:uc-free-steps}

In this section, we will show the calculation for finding the number of free steps run by the $\texttt{UC}$-decimation \texttt{DEC\textsubscript{UC}} on the random $k$-XORSAT instance $\mathbf{\Phi} \sim \mathbf{\Phi}_k(n,rn)$, which implies that \texttt{DEC\textsubscript{UC}} is $w_1(k,r)$-free.
To be specific, we will show the proof of Lemma \ref{lemma:uc-free}, by using the Wormald's method of differential equations \cite{wormaldDifferentialEquationsRandom1995}.

We start from approximating the degree profile of the factor graph of the remaining instance after $t$ iterations.
The degree profiles can help us to calculate the probability of the next iteration being a free step.
Note that the notations for describing the degree profile are defined in Appendix \ref{sec:term-degree-profile}.
For each of those notations, we append "$(t)$" to them to specify the values right after $t$ iterations.

\begin{lemma}
    \label{lemma:degree-profile}
    For any local rule $\tau$, after $t$ iterations of \texttt{DEC\textsubscript{$\tau$}}, w.h.p.the number of variables of degree $i$ is given by
    \begin{align}
        \label{eq:n_i}
        n_i(t) &= \binom{rn}{i} \left( \frac{k}{n} \right)^i \left( 1-\frac{k}{n} \right)^{rn-i} (n-t) + o(n), \text{\quad for $i=0,1,...,m$},
    \end{align}
    and the number of equations of degree $i$ is given by
    \begin{align}
        \label{eq:m_i}
        m_i(t) &= \binom{k}{i} \left(\frac{t}{n}\right)^{k-i} \left(1-\frac{t}{n}\right)^i rn + o(n), \text{\quad for $i=1,2,...,k$}.
    \end{align}
\end{lemma}

\begin{proof}
    We are going to track the changes in the numbers of variable nodes and equation nodes of different degrees throughout the process of the algorithm, by using the method of differential equations from \cite{wormaldDifferentialEquationsRandom1995}.

    To apply the method of differential equations, we need to compute the expected changes in the numbers of variable nodes and equation nodes of different degrees.
    Note that exactly one variable node is removed in each iteration, so the total number $\widehat{n}(t)$ of variable nodes is $n-t$ after $t$ iterations.

    Suppose we are at time $t \ge 0$ (i.e. after $t$ iterations).
    The selected variable node $x_{s(t)}$ is going to be removed from the factor graph $G_{\mathbf{\Phi_t}}$.
    The selected variable node $x_{s(t)}$ is of degree $i$ with probability $n_i(t)/\widehat{n}(t)$.
    So, for $i=0,1,2,...,m$, the expected change in the number $n_i$ of variable nodes of degree $i$ is $-1\times n_i(t)/\widehat{n}(t)$.
    This yields
    \begin{align}
        \expected{n_i(t+1)-n_i(t)} &= -\frac{n_i(t)}{\widehat{n}(t)} = -\frac{n_i(t)}{n-t} \label{eqt:change_in_ni}
    \end{align}
    for $i=0,1,...,m$.
    By Theorem \ref{thm:tree_like_neighborhood}, the local neighborhood $B_{\Phi_t}(x_{s(t)},R)$ of the selected variable $x_{s(t)}$ of radius $R$ in $G_{\mathbf{\Phi_t}}$ converges in distribution to $\mathbb{T}_R(\Lambda',P')$, where $\Lambda'_i=n_i(t)/\widehat{n}(t)$ for all $i$ and $P'_j=m_j(t)/\widehat{m}(t)$ for all $j$.
    Therefore, with probability $n_j/\widehat{n}$ for $j=0,1,2,...,m$, there are $j$ equation nodes directly adjacent to the selected variable node.
    For those equation nodes, the degree of each equation node is $l$ with probability $\rho_l(t)$, and decreases by 1 due to the removal of the selected variable node.
    In the other words, the expected changes in the numbers of equation nodes of different degree are given by
    \begin{align}
        \expected{m_k(t+1)-m_k(t)} &= \sum_{j=1}^m j\frac{n_j(t)}{\widehat{n}(t)} \cdot (\rho_k(t) \cdot (-1)) \nonumber \\
        &= \left( \sum_{j=1}^m j\frac{n_j(t)}{n-t} \right) \frac{-km_k(t)}{\sum_{j=1}^k jm_j(t)} \nonumber \\
        &= \left( \sum_{j=1}^m j\frac{n_j(t)}{n-t} \right) \frac{-km_k(t)}{\sum_{j=1}^m jn_j(t)} \nonumber \\
        &= \frac{-km_k(t)}{n-t} \label{eqt:change_in_mk}
    \end{align}
    and
    \begin{align}
        \expected{m_i(t+1)-m_i(t)} &= \sum_{j=1}^m j\frac{n_j}{\widehat{n}} \cdot [\rho_{i+1}(t)\cdot(+1)+\rho_i(t)\cdot(-1)] \nonumber \\
        &= \left( \sum_{j=1}^m j\frac{n_j(t)}{n-t} \right) \frac{(i+1)m_{i+1}(t)-im_{i}(t)}{\sum_{j=1}^k jm_j(t)} \nonumber \\
        &= \left( \sum_{j=1}^m j\frac{n_j(t)}{n-t} \right) \frac{(i+1)m_{i+1}(t)-im_{i}(t)}{\sum_{j=1}^m jn_j(t)} \nonumber \\
        &= \frac{(i+1)m_{i+1}(t)-im_{i}(t)}{n-t} \text{\quad for\;} i=1,2,...,k-1. \label{eqt:change_in_mi}
    \end{align}
    In the above calculation, we use the fact that $\sum_{j=1}^k jm_j(t) = \sum_{j=1}^m jn_j(t)$, which holds because both $\sum_{j=1}^k jm_j(t)$ and $\sum_{j=1}^m jn_j(t)$ are equal to the number of edges at time $t$.

    Let $x=t/n$ be the normalized time.
    Furthermore, let
    \begin{align*}
        y_i &= y_i(x) = m_i(xn)/n \text{\quad for\;} i=1,2,3,...,k, \\
        z_i &= z_i(x) = n_i(xn)/n \text{\quad for\;} i=0,1,2,...,m, \text{\quad and} \\
        z &= z(x)=\widehat{n}(xn)/n.
    \end{align*}
    Then, the equations (\ref{eqt:change_in_ni}), (\ref{eqt:change_in_mk}) and (\ref{eqt:change_in_mi}) suggest the following different equations:
    \begin{align}
        \frac{dz_i}{dx} &= -\frac{z_i}{1-x} &\text{for\;\;} i=0,1,2,...,m  \label{eq:differential_equation1} \\
        \frac{dy_k}{dx} &= \frac{-ky_k}{1-x} \label{eq:differential_equation2} \\
        \frac{dy_i}{dx} &= \frac{(i+1)y_{i+1}-iy_{i}}{1-x} &\text{for\;\;} i=1,2,3,...,k-1 \label{eq:differential_equation3}
    \end{align}

    At time $t=0$ (i.e. before running the algorithm), the number $n_i(0)$ of variable nodes of degree $i$ is $\binom{rn}{i} \left( \frac{k}{n} \right)^i \left( 1-\frac{k}{n} \right)^{rn-i} n + o(n)$ for all $i$.
    Since all $rn$ equation nodes are of degree $k$ at time $t=0$, $m_k(0)=rn$ and $m_i(0)=0$ for all $i$.
    These suggest the following initial conditions for those different equations above:
    \begin{align}
        z_i(0) &= \binom{rn}{i} \left( \frac{k}{n} \right)^i \left( 1-\frac{k}{n} \right)^{rn-i} &\text{for\;\;} i=0,1,2,...,m \label{eq:initial_condition1} \\
        y_k(0) &= r \label{eq:initial_condition2} \\
        y_i(0) &= 0 &\text{for\;\;} i=1,2,3,...,k-1 \label{eq:initial_condition3}
    \end{align}
    The solution of the different equations with the initial conditions above is given by
    \begin{align*}
        z_i(x) &= \binom{rn}{i} \left( \frac{k}{n} \right)^i \left( 1-\frac{k}{n} \right)^{rn-i} (1-x) &\text{for\;\;} i=0,1,2,...,m \\
        y_i(x) &= r \binom{k}{i} x^{k-1}(1-x)^i &\text{for\;\;} i=1,2,3,...,k
    \end{align*}
    for any $0 \le x \le 1$. (The steps of solving the differential equations are shown at Appendix \ref{sec:solving_differential_equations}.)
    By the Wormald's method of differential equations, at time $t \ge 0$, w.h.p. the number $n_i(t)$ of variable nodes of degree $i$ is equal to $z_i(t/n)n + o(n)$ for all $i$, and the number $m_i(t)$ of equation nodes of degree $i$ is equal to $y_i(t/n)n + o(n)$.
\end{proof}

In an iteration of the \texttt{UC}-decimation algorithm \texttt{DEC\textsubscript{UC}}, the local rule \texttt{UC} gives $1/2$ when there is no unit clause containing the selected variable $x_{s(t)}$.
Therefore, an iteration is a free step if and only if there is no equation node of degree 1 adjacent to the selected variable node $x_{s(t)}$.
With the degree profiles from Lemma \ref{lemma:degree-profile}, we can calculate the probability of an iteration being a free step, and thus approximate the total number of free steps.

\begin{proof}[Proof of Lemma \ref{lemma:uc-free}]
    We track the number of free steps after $t$ iterations, by using the Wormald's method of differential equations.
    Let $q(t)$ be the number of free steps after $t$ iterations.
    Note that $q(0)=0$.
    At time $t \ge 0$, if there exists at least one equation node of degree 1 adjacent to the selected variable node $x_{s(t)}$, then the $(t+1)$-st iteration is not a free step and $q(t+1)=q(t)$.
    Conversely, if all equation nodes adjacent to the selected variable node $x_{s(t)}$ are of degree not equal to 1, then the $(t+1)$-st iteration is a free step and $q(t+1)=q(t)+1$.
    The probability that all equation nodes adjacent to the selected variable node are of degree not equal to 1 is given by $\sum_{j=0}^{m} (n_j(t)/\widehat{n}(t)) (1-\rho_1(t))^i$.
    Therefore, we have
    \begin{align}
        \expected{q(t+1)-q(t)} &= \sum_{i=0}^{m} \frac{n_i(t)}{\widehat{n}(t)} (1-\rho_1(t))^i \cdot (+1) \nonumber \\
        &= \sum_{i=0}^{rn} \frac{n_i(t)}{n-t} \left( 1-\frac{m_1(t)}{\sum_{j=1}^{k} jm_j(t)} \right)^i.
    \end{align}
    From (\ref{eq:m_i}) and the polynomial identity $\sum_{i=1}^n i \binom{n}{i} (1-x)^i x^{n-i} = n(1-x)$, we can simplify $\sum_{j=1}^{k} jm_j(t)$ by
    \begin{align*}
        \sum_{j=1}^{k} jm_j(t)
        &= \sum_{j=1}^{k} j \binom{k}{i} \left(\frac{t}{n}\right)^{k-i} \left(1-\frac{t}{n}\right)^i rn + o(n) \\
        &= \left( k \cdot \left( 1-\frac{t}{n} \right) \right) rn + o(n) \\
        &= kr(n-t) + o(n).
    \end{align*}
    Then, by (\ref{eq:n_i}), (\ref{eq:m_i}) and the fact that $\lim_{n \rightarrow \infty} (1-\frac{kx^{k-1}}{n})^{rn} = e^{-krx^{k-1}}$, we have
    \begin{align}
        \expected{q(t+1)-q(t)}
        &= \sum_{i=0}^{m} \frac{n_i(t)}{n-t} \left( 1-\frac{\binom{k}{1} \left(\frac{t}{n}\right)^{k-1} \left(1-\frac{t}{n}\right) rn}{kr(n-t)} \right)^i + o(1) \nonumber \\
        &= \sum_{i=0}^{m} \frac{n_i(t)}{n-t} \left( 1-\left(\frac{t}{n}\right)^{k-1} \right)^i + o(1) \nonumber \\
        &= \sum_{i=0}^{m} n_i(0) \left( 1-\left(\frac{t}{n}\right)^{k-1} \right)^i + o(1) \nonumber \\
        &= \sum_{i=0}^{m} \binom{rn}{i} \left(\frac{k}{n}\right)^i \left(1-\frac{k}{n}\right)^{rn-i} \left( 1-\left(\frac{t}{n}\right)^{k-1} \right)^i + o(1) \nonumber \\
        &= \sum_{i=0}^{m} \binom{rn}{i} \left[\frac{k}{n} - \frac{k}{n}\left(\frac{t}{n}\right)^{k-1} \right]^i \left(1-\frac{k}{n}\right)^{rn-i} + o(1) \nonumber \\
        &= \left[ \frac{k}{n} - \frac{k}{n}\left(\frac{t}{n}\right)^{k-1} + 1 -\frac{k}{n} \right]^{rn} + o(1) \nonumber \\
        &= \left[ 1 - \frac{k}{n}\left(\frac{t}{n}\right)^{k-1} \right]^{rn} + o(1) \nonumber \\
        &= e^{-kr(t/n)^{k-1}} + o(1) \label{eqt:change_in_q}
    \end{align}
    By letting $x=t/n$ and $w = w(x) = q(xn)/n$, the equation (\ref{eqt:change_in_q}) and the fact that $q(0)=0$ suggest the differential equation
    \begin{align}
        \frac{dw}{dx} &= e^{-krx^{k-1}}
    \end{align}
    and the initial condition
    \begin{align}
        w(0)=0.
    \end{align}
    By the Wormald's method of differential equations, w.h.p. the number $q(t)$ of free steps after $t$ iterations is equal to $w(t/n)n + o(n)$.

    By the second fundamental theorem of calculus, we can write them as a definite integral
    \begin{align*}
        w(x) - w(0) = \int_{0}^{x} e^{-krt^{k-1}}dt
    \end{align*}
    Note that $\frac{d}{dt} (krt^{k-1}) = kr(k-1)t^{k-2}$.
    Using integration by substitution, we have
    \begin{align*}
        w(x)
        &= 0 + \int_{0}^{x} e^{-krt^{k-1}} dt \\
        &= \int_{0}^{x} e^{-krx^{k-1}} \frac{1}{kr(k-1)}t^{2-k} \cdot kr(k-1)t^{k-2} dt \\
        &= \int_{0}^{krx^{k-1}} e^{-krt^{k-1}} \frac{1}{kr(k-1)}t^{2-k} d(krt^{k-1}) \\
        &= \frac{(kr)^{\frac{1}{1-k}}}{k-1} \int_{0}^{krx^{k-1}} e^{-krt^{k-1}} (kr)^{\frac{2-k}{k-1}} t^{2-k} d(krt^{k-1}) \\
        &= \frac{(kr)^{\frac{1}{1-k}} }{k-1} \int_{0}^{krx^{k-1}} (krt^{k-1})^{\frac{1}{k-1}-1} e^{-krt^{k-1}} d(krt^{k-1}) \\
        &= \frac{(kr)^{\frac{1}{1-k}} }{k-1} \; \gamma \left( \frac{1}{k-1}, krx^{k-1} \right)
    \end{align*}
    where $\gamma$ is the lower incomplete gamma function defined by
    \begin{align*}
        \gamma(a,x) \equiv \int_{0}^{x} t^{a-1} e^{-t} dt.
    \end{align*}
    Hence, w.h.p. the total number of free steps run by the algorithm is $w_1(k,r)n + o(n)$, where
    \begin{align*}
        w_1(k,r) = \frac{(kr)^{\frac{1}{1-k}} }{k-1} \; \gamma \left( \frac{1}{k-1}, kr \right) = w(1).
    \end{align*}
    Hence, \texttt{DEC\textsubscript{UC}} is $w_1(k,r)$-free.
\end{proof}

\subsection{Solving differential equations (\ref{eq:differential_equation1}), (\ref{eq:differential_equation2}) and (\ref{eq:differential_equation3})}
\label{sec:solving_differential_equations}

In this section, we show how to solve the differential equations (\ref{eq:differential_equation1}), (\ref{eq:differential_equation2}) and (\ref{eq:differential_equation3}) with the initial conditions (\ref{eq:initial_condition1}), (\ref{eq:initial_condition2}) and (\ref{eq:initial_condition3}).
For $i=0,1,...,m$, the differential equations in (\ref{eq:differential_equation1}) can be written as
\begin{align*}
    \frac{1}{z_i} dz_i = \frac{-1}{1-x} dx.
\end{align*}
By separation of variable in integration and the initial condition (\ref{eq:initial_condition1}), we have
\begin{align}
    z_i(x) &= (1-x)z_i(0) \nonumber \\
    &= \binom{rn}{i} \left(\frac{k}{n}\right)^i \left(1-\frac{k}{n}\right)^{rn-j} (1-x) \text{\quad for all\;} i=0,1,...,m.
\end{align}

Next, we are going to use induction from $k$ to $1$ to prove that
\begin{align*}
    y_i(x) = r \binom{k}{i} x^{k-i}(1-x)^i
\end{align*}
First, the differential equation (\ref{eq:differential_equation2}) can be written as
\begin{align*}
    \frac{1}{y_k} dy_k = k\cdot \frac{-1}{1-x} dx
\end{align*}
By separation of variable in integration and the initial condition (\ref{eq:initial_condition2}), we have\
\begin{align*}
    y_k(x) = (1-x)^k y_k(0) = r (1-x)^k
\end{align*}
Now we assume $y_{i+1}(x) = r \binom{k}{i+1} x^{k-i-1}(1-x)^{i+1}$ for some $1 \le i < k$.
From the differential equations in (\ref{eq:differential_equation3}), we have
\begin{align*}
    \frac{dy_i}{dx} + \frac{i}{1-x}y_i = \frac{i+1}{1-x}y_{i+1}
\end{align*}
By applying the standard method for the first order linear differential equations and the induction assumption, we have 
\begin{align*}
    y_i(x) 
    &= \frac{ \int \mu(x) \frac{i+1}{1-x}y_{i+1} dx + c }{\mu(x)} \\
    &= \frac{ \int \frac{1}{(1-x)^i} \frac{i+1}{1-x} r \binom{k}{i+1} x^{k-i-1}(1-x)^{i+1} dx + c }{ \frac{1}{(1-x)^i}} \\
    &= \frac{i+1}{k-i} r \binom{k}{i+1} x^{k-1} (1-x)^i + c(1-x)^i \\
    &= r \binom{k}{i} x^{k-1} (1-x)^i + c(1-x)^i
\end{align*}
where $c$ is a constant and $\mu(x) = \exp \left( \int \frac{i}{1-x} dx \right) = \frac{1}{(1-x)^i}$.
Since $1 \le i < k$, $y_i(0)=0$ and thus $c=0$. Hence, we have
\begin{align}
    y_i(x) = r \binom{k}{i} x^{k-i} (1-x)^i,
\end{align}
which completes the induction.

\section{Proof of Lemma \ref{lemma:bp-free}: Finding the number of free steps in \texttt{DEC$_\tau$}}
\label{sec:bp-free-steps}

Now we assume the local rule $\tau$ can give the marginal probabilities of variables for any factor graph that is a tree.
We then show the calculation for the number of free steps run by the $\tau$-decimation algorithm \texttt{DEC$_\tau$} on the random $k$-XORSAT instance $\mathbf{\Phi} \sim \mathbf{\Phi}_k(n,rn)$.
To be specific, we show the proof of Lemma \ref{lemma:bp-free}, by using the Wormald's method of differential equations \cite{wormaldDifferentialEquationMethod1999}.

In this section, the approach is slightly different from the proof in Appendix \ref{sec:uc-free-steps} since we do not have the details on how the local rule $\tau$ calculates its output value.
However, since we assume that the local rule $\tau$ can give the exact marginals, we can determine the value output by the local rule $\tau$, by studying the structure of the local neighborhood $B_{\mathbf{\Phi_t}}(x_{s(t)},R)$ at time $t \ge 0$.

Assume we are at time $t \ge 0$.
We consider the local neighborhood $B_{\mathbf{\Phi_t}}(x_{s(t)},R)$ of the selected variable $x_{s(t)}$ of radius $R \ge 0$.
By Theorem \ref{thm:tree_like_neighborhood}, the local neighborhood $B_{\mathbf{\Phi_t}}(x_{s(t)},R)$ converges in distribution to the tree ensemble $\mathbb{T}_R(\Lambda', P')$, where $(\Lambda', P')$ is the degree distribution of the factor graph $\mathbf{\Phi_t}$ at time $t$.
The values of $\Lambda_i$ and $P_i$ are given by (\ref{eq:n_i}) and (\ref{eq:m_i}) since Lemma \ref{lemma:degree-profile} can be applied to $\tau$-decimation with any local rule $\tau$.
By our assumption, the value output by the local rule $\tau$ is equal to the marginal probability of the selected variable $x_{s(t)}$ on a random solution for the XORSAT instance induced by the local neighborhood $B_{\mathbf{\Phi_t}}(x_{s(t)},R)$.
So, we can obtain the value output by the local rule by calculating the marginal probability of $x_{s(t)}$, using the tree ensemble $\mathbb{T}_R(\Lambda',P')$.

For any $0 \le l \le R$, we denote by $T_l(x')$ a subtree of $B_{\mathbf{\Phi_t}}(x_{s(t)},R)$ rooted at some node $x'$ whose distance from $x_{s(t)}$ is $R-l$.
According to the tree ensemble $\mathbb{T}_R(\Lambda',P')$, the tree $T_l(x)$ is i.i.d. for all variable nodes $x'$ of same distance $R-l$ from $x_{s(t)}$.
By abusing the notation, we can simply omit "$(x)$" and write $T_l=T_l(x)$.
For any factor graph $T$ which is a tree rooted at a variable node $x'$, we say the tree $T$ \textbf{has a free root} if in the XORSAT instance induced by $T$ the marginal distribution of the root variable $x'$ is an even distribution over $\{0,1\}$.
Now, we are going to prove that for any $0 \le l \le R-1$ the tree $T_l$ has a free root with probability at least $S_l(t/n) + o(1)$, where the sequence $\{S_l(x)\}_{l \ge 0}$ is given by
\begin{align}
    S_0(x) = 1 \text{,\quad and \quad}
    S_{l}(x) = \exp \left( -kr \left[ \left( 1-x \right)(1-S_{l-1}(x)) + x \right]^{k-1} \right) \text{\quad for any\;}  l \ge 1
\end{align}
for any $x \in \mathbb{R}$.

\begin{lemma}
    \label{lemma:T_l_property_F}
    For $0 \le l \le R-1$, the tree $T_l$ has a free root with probability at least $S_l(t/n) + o(1)$.
\end{lemma}

\begin{proof}
    First, we consider the tree $T_0$.
    The root variable node $x'$ of $T_0$ has distance $R$ from the selected variable $x_{s(t)}$.
    Since the radius of $B_{\mathbf{\Phi_t}}(x_{s(t)},R)$ is $R$, the variable $x'$ does not have any child node.
    Thus, $T_0$ consists of only one variable node, namely $x'$, and no equation node.
    We can assign either 0 or 1 to $x$ without violating any equation.
    Hence, the tree $T_0$ has a free root with probability $S_0(t/n)=1$.
    
    For $1 \le l \le R-1$, consider the subtree $T_{l}$ of local neighborhood $B_{\mathbf{\Phi_t}}(x_{s(t)},R)$, rooted at the variable node $x_a$ of distance $R-l$ from $x_{s(t)}$.
    The variable node $x_a$ has $i-1$ child equation nodes with probability $\lambda_i(t)$ for $1 \le i \le m$.
    Each of those equation nodes $e_a$ has $j-1$ child variable nodes $x_b$ with probability $\rho_j(t)$ for $1 \le j \le k$.
    We also know that each of these child variable nodes $x_b$ is the root of an i.i.d. subtree $T_{l-1}$.
    In other words, each of those equation nodes $e_a$ is connected to the roots of $j-1$ i.i.d. subtrees $T_{l-1}$ as its children.
    
    For an equation node $e_a$ mentioned above, if at least one of its child subtree $T_{l-1}$ has a free root, then there are at least two solutions for the subtree $T_{l-1}$, one assigns 0 to the root $x_b$ of subtree $T_{l-1}$, and another one assigns 1 to the root $x_b$ of subtree $T_{l-1}$.
    Therefore, no matter what value we assign to $x_a$, we are able to choose a suitable assignment for the variables in that subtree $T_{l-1}$ so that the equation of $e_a$ and all equations in $T_{l-1}$ are satisfied. 

    Note that, from (\ref{eq:n_i}) and (\ref{eq:m_i}), we know that at time $t$
    \begin{align*}
        \lambda_i(t) &= i \binom{rn}{i} \left( \frac{k}{n} \right)^{i} \left( 1 - \frac{k}{n} \right)^{rn-i} \frac{1}{kr} + o(1) \text{\quad for $1 \le i \le m$, and} \\
        \rho_j(t) &= j \binom{k}{j} \left( \frac{t}{n} \right)^{k-j} \left( 1 - \frac{t}{n} \right)^j \frac{1}{k \left( 1-\frac{t}{n} \right)} + o(1) \text{\quad for $1 \le j \le k$.}
    \end{align*}
    So, the probability of the equation node $e_a$ having at least one of its child subtree $T_{l-1}$ having a free root is given by
    \begin{align*}
        &1 - \sum_{j=1}^{k} \rho_j(t) (1-S_{l-1})^{j-1} \\
        &= 1 - \sum_{j=1}^{k} j \binom{k}{j} \left( \frac{k}{n} \right)^{k-j} \left( 1 - \frac{t}{n} \right)^j (1-S_{l-1})^{j} \cdot \frac{1}{k \left( 1-\frac{t}{n} \right) (1-S_{l-1})} + o(1) \\
        &= 1 - \sum_{j=1}^{k} j \binom{k}{j} \left( \frac{k}{n} \right)^{k-j} \left[ \left( 1 - \frac{t}{n} \right) (1-S_{l-1}) \right]^j \cdot \frac{1}{k \left( 1-\frac{t}{n} \right) (1-S_{l-1})} + o(1) \\
        &= 1 - k \left[ \left( 1-\frac{t}{n} \right) (1-S_{l-1}) + \frac{t}{n} \right]^{k-1} \left( 1-\frac{t}{n} \right) (1-S_{l-1}) \\
        & \quad\quad\quad\quad\quad\quad\quad\quad\quad\quad\quad\quad\quad\quad\quad\quad\quad\quad\quad\quad\quad \cdot \frac{1}{k \left( 1-\frac{t}{n} \right) (1-S_{l-1})} + o(1) \\
        &= 1 - \left[ \left( 1-\frac{t}{n} \right) (1-S_{l-1}) + \frac{t}{n} \right]^{k-1} + o(1) \\
        &\equiv S_l^*,
    \end{align*}
    where $S_{l-1}=S_{l-1}(t/n)$.
    Note that $\lim_{n \rightarrow \infty} (1+\alpha/n)^n = e^\alpha$ for all $\alpha \in \mathbb{R}$.
    Then, the subtree $T_{l}$ has a free root with probability
    \begin{align*}
        \sum_{i=1}^{rn} \lambda_i(t) (S_l^*)^{i-1}
        &= \sum_{i=1}^{rn} i \binom{rn}{i} \left( \frac{k}{n} \right)^{i} \left( 1 - \frac{k}{n} \right)^{rn-i} (S_l^*)^{i} \cdot \frac{1}{krS_l^*} + o(1) \\
        &= \sum_{i=1}^{rn} i \binom{rn}{i} \left( \frac{k}{n} S_l^* \right)^{i} \left( 1 - \frac{k}{n} \right)^{rn-i} \cdot \frac{1}{krS_l^*} + o(1) \\
        &= rn \left[ \frac{k}{n}S_l^* + \left( 1 - \frac{k}{n} \right) \right]^{rn-1} \left( \frac{k}{n} S_l^* \right) \cdot \frac{1}{krS_l^*} + o(1) \\
        &= \left[ \frac{k}{n}S_l^* + \left( 1 - \frac{k}{n} \right) \right]^{rn-1} + o(1) \\
        &= \left( 1 + \frac{k(S_l^*-1)}{n} \right)^{rn-1} + o(1) \\
        &= \exp(kr(S_l^*-1)) + o(1) \\
        &= \exp\left( -kr \left[ \left( 1-\frac{t}{n} \right) (1-S_{l-1}) + \frac{t}{n} \right]^{k-1} \right) + o(1) \\
        &= S_{l}(t/n) + o(1).
    \end{align*}
\end{proof}

Now, we can calculate the probability that the local neighborhood $B_{\mathbf{\Phi_t}}(x_{s(t)},R)$ of the selected variable $x_{s(t)}$ of radius $R>0$ has a free root with probability at least $S_R(t/n)$.
The proof is similar to the proof of Lemma \ref{lemma:T_l_property_F}, except replacing $\lambda_i(t)$ with $\Lambda_i(t)$.
\begin{lemma}
    \label{lemma:B_property_F}
    At time $t \ge 0$, the local neighborhood $B_{\mathbf{\Phi_t}}(x_{s(t)},R)$ of the selected variable $x_{s(t)}$ of radius $R \ge 0$ has a free root with probability at least $S_R(t/n) + o(1)$.
\end{lemma}

\begin{proof}
    The root variable node $x_{s(t)}$ has $i$ child equation node with probability $\Lambda_i(t)$ for $0 \le i \le m$.
    Each of those equation nodes $e_a$ has $j-1$ child variable nodes $x_b$ with probability $\rho_j(t)$, and each of these child variable nodes $x_b$ is the root of an i.i.d. subtree $T_{R-1}$.
    In other words, each of those equation nodes $e_a$ is connected to the roots of $j-1$ i.i.d. subtrees $T_{R-1}$ as its children.

    For an equation node $e_a$ mentioned above, if at least one of its child subtree $T_{R-1}$ has a free root, then we are able to obtain a satisfying assignment for all variable nodes in the child subtree $T_{R-1}$ which assigns either 1 or 0 to the root $x_b$ of subtree $T_{R-1}$.
    Therefore, no matter what value we assign to $x_a$, we are able to choose a suitable assignment for the variables in that subtree $T_{R-1}$ so that the equation of $e_a$ and all equations in $T_{R-1}$ are satisfied.

    Similar to the proof of Lemma \ref{lemma:T_l_property_F}, the probability that the equation node $e_a$ has at least one of its child subtree $T_{R-1}$ having a free root is given by
    \begin{align*}
        1 - \sum_{j=1}^k \rho_j(t) (1-S_{R-1})^{j-1}
        &= \exp\left( -kr \left[ \left( 1-\frac{t}{n} \right) (1-S_{R-1}) + \frac{t}{n} \right]^{k-1} \right) + o(1)
        \equiv S_{R}^*,
    \end{align*}
    where $S_{R-1}=S_{R-1}(t/n)$.
    From (\ref{eq:n_i}), at time $t$ we have
        $\Lambda_i(t) = \binom{rn}{i} \left( \frac{k}{n} \right)^{i} \left( 1 - \frac{k}{n} \right)^{rn-i} + o(1)$
    for $0 \le i \le m$.
    Note that $\lim_{n \rightarrow \infty} (1+\alpha/n)^n = e^\alpha$ for all $\alpha \in \mathbb{R}$.
    Hence, the probability that the local neighborhood $B_{\mathbf{\Phi_t}}(x_{s(t)},R)$ has a free root is given by
    \begin{align*}
        \sum_{i=0}^{rn} \Lambda_i(t) (S_{R}^*)^{i-1}
        &= \sum_{i=0}^{rn} \binom{rn}{i} \left( \frac{k}{n} \right)^{i} \left( 1 - \frac{k}{n} \right)^{rn-i} (S_{R}^*)^{i} + o(1) \\
        &= \sum_{i=0}^{rn} \binom{rn}{i} \left( \frac{k}{n}S_{R}^* \right)^{i} \left( 1 - \frac{k}{n} \right)^{rn-i} + o(1) \\
        &= \left[ \frac{k}{n} S_{R}^* + \left( 1-\frac{k}{n} \right) \right]^{rn} + o(1) \\
        &= \left[ 1 + \frac{k(S_{R}^*-1)}{n} \right]^{rn} + o(1) \\
        &= \exp \left( kr(S_{R}^*-1) \right) + o(1) \\
        &= \exp \left( -kr \left[ \left( 1-\frac{t}{n} \right) (1-S_{R-1}) + \frac{t}{n} \right]^{k-1} \right) + o(1) \\
        &= S_{R}(t/n) + o(1),
    \end{align*}
\end{proof}

\begin{proof}[Proof of Lemma \ref{lemma:bp-free}]
    Lemma \ref{lemma:B_property_F} implies that at time $t \ge 0$ the local rule $\tau$ gives the value $1/2$ to the decimation algorithm with probability at least $S_R(t/n)$.
    Now we can calculate the number of free steps run by the $\tau$-decimation algorithm by tracking the number throughout the process of the algorithm.
    We know that the $(t+1)$-st iteration is a free step with probability at least $S_R(t/n) + o(1)$.
    Let $q(t)$ be a lower bound of the number of free steps after $t$ iterations, with $q(0) = 0$ and
    \begin{align}
        \expected{q_e(t+1) - q_e(t)} = S_R(t/n) + o(1). \label{eq:change_in_q_bp}
    \end{align}
    By letting $x=t/n$ and $w = w(x) = q(xn)/n$, the equation (\ref{eq:change_in_q_bp}) suggests the differential equation
    \begin{align}
        \frac{dw}{dx} = S_R(x)
    \end{align}
    with the initial condition
    \begin{align}
        w(x) = 0
    \end{align}
    since $q(0)=0$.
    By the Wormald's method of differential equations, w.h.p. the lower bound $q(t)$ of the number of free steps after $t$ iterations is given by
    \begin{align*}
        \left( \int_{0}^{t/n} S_R(x) dx \right) n + o(n).
    \end{align*}
    Hence, w.h.p. the total number of free steps run by the $\tau$-decimation algorithm is lower bounded by $w_e(k,r)n + o(n)$, where $w_e(k,r)$ is given by
    \begin{align*}
        w_e(k,r) = \int_{0}^{1} S_R(x) dx,
    \end{align*}
    Hence, \texttt{DEC$_\tau$} is $w_e(k,r)$-free.
\end{proof}

\section{Proof of Lemma \ref{lemma:lower-bound-w-1}}
\label{sec:proof-of-lower-bound-w-1}

In this section, we prove Lemma \ref{lemma:lower-bound-w-1}, which gives a lower bound for $w_1(k,r)$.
To do that, we first prove that $w_1(k,r)$
is lower bounded by
\begin{align*}
    w_1^*(k) = \frac{k^{\frac{1}{1-k}}}{k-1} \gamma \left( \frac{1}{k-1}, k \left(\frac{k}{k+1}\right)^{k-1} \right),
\end{align*}
for $k \ge 3$ and $r \in [0,1]$.
We then prove that $w_1^*(k)$ is decreasing with integer $k \ge 3$, and thus $w_1(k,r)$ has a lower bound $w_1^*(k_0)$ for any $k \ge k_0 \ge 3$.

We start from proving that $w_1(k,r)$ is decreasing with $r \in [0,1]$, which implies $w_1(k,r) \ge w_1(k,1)$ for any $r \in [0,1]$ and $k \ge 3$.
\begin{lemma}
    \label{lemma:w1_decreasing_with_r}
    $w_1(k,r)$ is decreasing with $r \in [0,1]$ for any $k \ge 3$.
\end{lemma}
\begin{proof}
    We obtain the derivative of $w_1(k,r)$ with respect to $r$ by the followings.
    \begin{align*}
        \frac{\partial w_1}{\partial r}
        &= \frac{\frac{1}{1-k} (kr)^{\frac{1}{1-k}-1}k}{k-1} \gamma \left( \frac{1}{k-1}, kr \right) + \frac{(kr)^{\frac{1}{1-k}}}{k-1} (kr)^{\frac{1}{1-k}-1} e^{-kr} k \\
        &= \frac{k(kr)^{\frac{1}{1-k}-1}}{k-1} \left[ \frac{1}{1-k} \gamma\left( \frac{1}{k-1}, kr \right) + (kr)^{\frac{1}{k-1}} e^{-kr} \right] \\
        &= \frac{k(kr)^{\frac{1}{1-k}-1}}{k-1} h_k(r)
    \end{align*}
    where $h_k(r)$ is given by
    \begin{align*}
        h_k(r) = \frac{1}{1-k} \gamma\left( \frac{1}{k-1}, kr \right) + (kr)^{\frac{1}{k-1}} e^{-kr}.
    \end{align*}
    Note that $h_k(0) = 0$ and its derivative is given by
    \begin{align*}
        \frac{dh_k}{dr}
        &= \frac{1}{1-k} (kr)^{\frac{1}{k-1}-1} e^{-kr} k + \frac{1}{k-1} (kr)^{\frac{1}{k-1}-1} k e^{-kr} + (kr)^{\frac{1}{k-1}} e^{-kr} (-k) \\
        &= -k (kr)^{\frac{1}{k-1}} e^{-kr} \\
        &\le 0.
    \end{align*}
    Therefore, $h_k(r)$ is decreasing with $r$ and $h_k(r) \le h_k(0) = 0$ for any $r \in [0,1]$.
    Hence, $\frac{\partial w_1}{\partial r} \le 0$ and thus $w_1$ is decreasing with $r \in [0,1]$ for any $k \ge 3$.
\end{proof}

By the above lemma, we know that $w_1(k,r) \ge w_1(k,1)$ for any $r \in [0,1]$.
Next we will prove that $w_1(k,1)$ is lower bounded by $w_1^*(k)$.
\begin{lemma}
    \label{lemma:w1_and_w1_star}
    For $k \ge 3$, $w_1(k,1) \ge w_1^*(k)$.
\end{lemma}
\begin{proof}
    For any $k \ge 3$ and any $t \ge 0$ we have
    \begin{align*}
        k > k \left( \frac{k}{k+1} \right)^{k-1} \text{\quad and \quad} t^{\frac{1}{k-1}-1} e^{-t} > 0,
    \end{align*}
    This implies that
    \begin{align}
        w_1(k,1)
        = \frac{k^{\frac{1}{1-k}}}{k-1} \gamma \left( \frac{1}{k-1}, k \right)
        > \frac{k^{\frac{1}{1-k}}}{k-1} \gamma \left( \frac{1}{k-1}, k \left(\frac{k}{k+1}\right)^{k-1} \right)
        = w_1^*(k).
    \end{align}
\end{proof}

Next, we prove that $\{w_1^*(k)\}_{k \ge 3}$ is an increasing sequence, which implies that $w_1^*(k) \ge w_1^*(k_0)$ for any $k \ge k_0$.
\begin{lemma}
    \label{lemma:w1_increasing_with_k}
    $\{w_1^*(k)\}_{k \ge 3}$ is an increasing sequence.
\end{lemma}
\begin{proof}
    Note that for $k, x > 0$ we have
    \begin{align*}
        \int e^{-kt^{k-1}} dx
        &= - \frac{k^{\frac{1}{1-k}}}{k-1} \; \Gamma \left( \frac{1}{k-1}, kx^{k-1} \right) + \text{constant.}
    \end{align*}
    By replacing $k$ with $k+1$, we also have
    \begin{align*}
        \int e^{-(k+1)t^{k}} dx
        &= - \frac{(k+1)^{\frac{1}{-k}}}{k} \; \Gamma \left( \frac{1}{k}, (k+1)x^{k} \right) + \text{constant.}
    \end{align*}
    Therefore, we have
    \begin{align*}
        \int_{0}^{\frac{k}{k+1}} e^{-kt^{k-1}} dx
        &= \frac{k^{\frac{1}{1-k}}}{k-1} \; \left[ \Gamma \left( \frac{1}{k-1} \right) - \Gamma \left( \frac{1}{k-1}, k\left( \frac{k}{k+1} \right)^{k-1} \right) \right] \\
        &= \frac{k^{\frac{1}{1-k}}}{k-1} \; \gamma \left( \frac{1}{k-1}, k\left( \frac{k}{k+1} \right)^{k-1} \right) \\
        &= w_1^*(k)
    \end{align*}
    and 
    \begin{align*}
        \int_{0}^{\frac{k}{k+1}} e^{-(k+1)t^{k}} dx
        &= \frac{(k+1)^{\frac{1}{-k}}}{k} \; \left[ \Gamma \left( \frac{1}{k} \right) - \Gamma \left( \frac{1}{k}, (k+1)\left( \frac{k}{k+1} \right)^{k} \right) \right] \\
        &= \frac{(k+1)^{\frac{1}{-k}}}{k} \; \gamma \left( \frac{1}{k}, (k+1)\left( \frac{k}{k+1} \right)^{k} \right) \\
        &< \frac{(k+1)^{\frac{1}{-k}}}{k} \; \gamma \left( \frac{1}{k}, (k+1)\left( \frac{k+1}{k+2} \right)^{k} \right) \\
        &= w_1^*(k+1).
    \end{align*}
    The above inequality is based on the fact that the lower incomplete gamma function $\gamma(a,x)$ is strictly increasing with $x \ge 0$ and $\frac{k}{k+1} < \frac{k+1}{k+2}$.

    For $0 \le x \le \frac{k}{k+1}$, we have $e^{-kx^{k-1}} \le e^{-(k+1)x^k}$.
    Therefore, we have
    \begin{align*}
        w_1^*(k)
        = \int_{0}^{\frac{k}{k+1}} e^{-kt^{k-1}} dx
        \le \int_{0}^{\frac{k}{k+1}} e^{-(k+1)t^{k}} dx
        \le  w_1^*(k+1)
    \end{align*}
    and thus $\{w_1^*(k)\}_{k \ge 3}$ is an increasing sequence.
\end{proof}

Combining above lemmas, we can complete the proof of Lemma \ref{lemma:lower-bound-w-1}.
\begin{proof}[Proof of Lemma \ref{lemma:lower-bound-w-1}]
    From Lemma \ref{lemma:w1_decreasing_with_r}, \ref{lemma:w1_and_w1_star} and \ref{lemma:w1_increasing_with_k},
    we have $w_1(k,r) \ge w_1(k,1) \ge w_1^*(k) \ge w_1^*(k_0)$ for any $k \ge k_0 \ge 3$ and $r \in [0,1]$.
\end{proof}

\section{Proof of Lemma \ref{lemma:lower-bound-w-b}}
\label{sec:proof-of-lower-bound-w-b}

In this section, we prove Lemma \ref{lemma:lower-bound-w-b}, which gives a lower bound for $w_e(k,r)$.
Recall that
\begin{align}
    S_0(x)=1 \text{\;\;and\;\;} S_l(x) = \exp \left( -kr [ (1-x)(1-S_{l-1}(x)) + x]^{k-1} \right) \text{\quad for any\;} l \ge 1 \text{\;and\;} x \in \mathbb{R}.
    \label{eq:def-sl}
\end{align}
We first show that $\{S_l(x)\}_{l\ge 0}$ is decreasing.

\begin{lemma}
    \label{lemma:s_l_decreasing}
    For any $k \ge 3$, $r \in [0,1]$ and $x \in [0,1]$, we have $0 < S_l(x) \le 1$ for any $l \ge 0$, and the sequence $\{S_l(x)\}_{l \ge 0}$ is decreasing.
\end{lemma}
\begin{proof}
    Without loss of generality, we write $S_l=S_l(x)$.
    First, we prove that $0 < S_l \le 1$ for all $l \ge 0$, by induction.
    Note that $S_0 = 1$.
    If $0 < S_l \le 1$ for some $l \ge 0$, then it is easy to see from (\ref{eq:def-sl}) that we also have $0 < S_{l+1} \le 1$.
    Therefore, we know that $0 < S_l \le 1$ for all $l \ge 0$.

    Then, we prove $S_{l+1} \le S_{l}$ for any $l \ge 0$ by induction again.
    For $l=0$, we have $S_1 = \exp(-kr[(1-x)(1-S_0)+x]^{k-1}) = \exp(-krx^{k-1}) \le 1 = S_0$.
    Assume that $S_{l+1} \le S_{l}$ for some $l \ge 0$.
    Then, we have
    \begin{align*}
        S_{l+2} 
        &= \exp ( -kr [(1-x)(1-S_{l+1})+x]^{k-1} ) \\
        &\le \exp ( -kr [(1-x)(1-S_{l})+x]^{k-1} ) \\
        &= S_{l+1}.
    \end{align*}
    Therefore, $S_{l+1} \le S_{l}$ is true for all $l \ge 0$.
    The result follows.
\end{proof}

Since the sequence $\{S_l(x)\}_{l \ge 0}$ is decreasing and bounded from below by $0$, by monotone convergence theorem, it converges as $l \rightarrow \infty$.
In particular, $\{S_l(x)\}_{l \ge 0}$ converges to $\hat{S}(x)$ as $l \rightarrow \infty$, where $\hat{S}(x)$ is the largest solution of the fixed point equation
\begin{align}
    \label{eq:s_fixed_point_equation}
    S = \exp ( -kr [(1-x)(1-S) + x]^{k-1} ),
\end{align}
and $S_l(x) \ge \hat{S}(x)$ for any $l \ge 0$.

\begin{lemma}
    \label{lemma:hat_s_lower_bound}
    Given $k \ge 3$ and $r \in [0,1]$, we have $\hat{S}(x) \ge 1-(kr)^2 x^{k-1}$ for any $0 \le x < x^-(k,r)$, where
    \begin{align}
        \label{eq:x_pm}
        x^\pm(k,r) = \left( \frac{1 \pm \sqrt{1-4(kr)^{-2}[(kr)^{\frac{1}{k-1}}-1]}}{2} \right)^{\frac{1}{k-2}}.
    \end{align}
\end{lemma}

\begin{proof}
    Define the analytic real-valued function $F(k,r,x,s)$ by
    \begin{align*}
        F(k,r,x,s) = \exp( -kr [(1-x)(1-s)+x]^{k-1})-s
    \end{align*}
    for any $k \ge 3$, $r \in [0,1]$, $x \in [0,1]$ and $s \in \mathbb{R}$.
    Note that we have
    \begin{align*}
        F(k,r,x,s)
        &= \exp( -kr [(1-x)(1-s)+x]^{k-1} ) - s \\
        &\ge 1 - kr [(1-x)(1-s)+x]^{k-1} - s \\
        &\ge 1 - kr [(1-x^{k-2})(1-s)+x]^{k-1} - s
    \end{align*}
    since $\exp(-y) \ge 1 - y$ for any $y \in \mathbb{R}$, and $x \in [0,1]$.
    Now we set $s = 1-(kr)^2 x^{k-1}$.
    With the polynomial identity $X^{k-1}-Y^{k-1} = (X-Y)\sum_{i=1}^{k-1} X^{k-1-i}Y^{i-1}$, we then have
    \begin{align*}
        &F(k,r,x,1-(kr)^2 x^{k-1}) \\
        &\ge (kr)^2 x^{k-1} - kr [(1-x^{k-2})(kr)^2x^{k-1}+x]^{k-1} \\
        &= kr \left( [(kr)^{\frac{1}{k-1}} x]^{k-1} - [(1-x^{k-2})(kr)^2x^{k-1}+x]^{k-1} \right) \\
        &= kr \left( [(kr)^{\frac{1}{k-1}} x] - [(1-x^{k-2})(kr)^2x^{k-1}+x] \right) \cdot F_2(k,r,x) \\
        &= kr \left( (kr)^{\frac{1}{k-1}} x - (kr)^2x^{k-1} + (kr)^2x^{2k-3} - x \right) \cdot F_2(k,r,x) \\
        &= krx \left( (kr)^2(x^{k-2})^2 - (kr)^2(x^{k-2}) + ((kr)^{\frac{1}{k-1}} - 1) \right) \cdot F_2(k,r,x) \\
        &= krx \cdot F_1(k,r,x) \cdot F_2(k,r,x)
    \end{align*}
    where
    \begin{align*}
        F_1(k,r,x) &= (kr)^2(x^{k-2})^2 - (kr)^2(x^{k-2}) + ((kr)^{\frac{1}{k-1}} - 1) \text{\quad and} \\
        F_2(k,r,x) &= \sum_{i=1}^{k-1} [(kr)^{\frac{1}{k-1}} x]^{k-1-i} [(1-x^{k-2})(kr)^2x^{k-1}+x]^{i-1}.
    \end{align*}
    By the quadratic formula and (\ref{eq:x_pm}), we know that $F_1(k,r,x) \ge 0$ if $x \le x^-(k,r)$ or $x \ge x^+(k,r)$.
    In particular, $F_1(k,r,x) = 0$ if $x = x^-(k,r)$ or $x = x^+(k,r)$.
    For $x \ge 0$, we have $(kr)^{\frac{1}{k-1}}x \ge 0$ and $(1-x)^{k-2}(kr)^2x^{k-1}+x \ge 0$, and thus $F_2(k,r,x) \ge 0$.
    Hence, we have $F(k,r,x,1-(kr)^2x^{k-1}) \ge 0$ for any $0 \le x \le x^-(k,r)$.

    Now we view $F$ as a function of $s$.
    For $0 \le x \le x^-(k,r)$, we know that $F(k,r,x,1-(kr)^2x^{k-1}) \ge 0$ and $F(k,r,x,1) = \exp(-krx^{k-1}) - 1 \le 0$.
    By the continuity of $F$ as a function of $s$, there exists at least one $s_0 \in [1-krx^{k-1},1]$ such that $F(k,r,x,s_0) = 0$, which implies
    \begin{align*}
        s_0 = \exp( -kr [(1-x)(1-s_0)+x]^{k-1} ).
    \end{align*}
    By definition, $\hat{S}(k,r,x)$ is the largest solution of the fixed point equation (\ref{eq:s_fixed_point_equation}), so we have $\hat{S}(x) \ge s_0 \ge 1-krx^{k-1}$.
\end{proof}

\begin{proof}[Proof of Lemma \ref{lemma:lower-bound-w-b}]
    From Lemma \ref{lemma:s_l_decreasing}, with $x \in (r_{core}(k),r_{sat}(k))$, we know the sequence $\{S_l(x)\}_{l \ge 0}$ is decreasing and lower bounded by 0.
    By the monotone convergence theorem, the sequence $\{S_l(x)\}_{l \ge 0}$ converges as $l \rightarrow \infty$.
    In particular, it converges to $\hat{S}(x)$ which is the largest solution of the fixed point equation in (\ref{eq:s_fixed_point_equation}), and $S_l(x) \ge \hat{S}(x)$ for all $l \ge 0$.
    Therefore, we have
    \begin{align*}
        w_e(k,r)
        = \int_{0}^{1} S_R(x) dx
        \ge \int_{0}^{1} \hat{S}(x) dx.
    \end{align*}
    Furthermore, since $\{S_l(x)\}_{l \ge 0}$ is lower bounded by $0$, $\hat{S}(x)$ is non-negative for any $x \in [0,1]$.
    Thus, we have 
    \begin{align*}
        \int_{0}^{1} \hat{S}(x) dx
        \ge \int_{0}^{x^-(k,r)} \hat{S}(x) dx.
    \end{align*}
    Then, by Lemma \ref{lemma:hat_s_lower_bound}, we have $\hat{S}(x) \ge 1 - (kr)^2x^{k-1}$ for any $0 \le x \le x^-(k,r)$, which implies
    \begin{align*}
        \int_{0}^{1} \hat{S}(x) dx
        &\ge \int_{0}^{x^-(k,r)} (1 - (kr)^2x^{k-1}) dx \\
        &\ge x^-(k,r) - kr^2 (x^-(k,r))^{k} \\
        &= w_e^*(k,r).
    \end{align*}
    By directly differentiating $w_e^*(k,r)$ with respect to $k$ and $r$, we can check that $w_e^*(k,r)$ is increasing with $k$ for $k \ge 3$ and decreasing with $r$ for $r \in (r_{core}(k),r_{sat}(k))$.
    Therefore, we have $w_e(k,r) \ge w_e^*(k_0,r_{sat}(k_0))$.
\end{proof}

\section{Proof of Lemma \ref{lemma:r-1}}
\label{sec:proof-of-ogp-of-whole-instance}

The natural way to reveal the overlap gap property of the random $k$-XORSAT problem is to use the first moment method.
Given a random $k$-XORSAT instance $\mathbf{\Phi} \sim \mathbf{\Phi}_n(k,rn)$ of $n$ variables and $rn$ clauses, we first calculate the expected number of pairs of solutions with distance $\alpha n$ between them, for any $\alpha \in [0,1]$.
Let $Z(\alpha n)$ be the number of pairs of solutions with distance $\alpha n$ between them.
The expected value of $Z(\alpha n)$ is given by the following lemma.

\begin{lemma}
    \label{lemma:first-moment-method}
    Let $k \ge 3$ and $r > 0$.
    For any $\alpha \in [0,1]$, the expected value of $Z(\alpha n)$ is given by
    \begin{align*}
        \expected{Z(\alpha n)} = \frac{1}{\sqrt{2\pi n}} \frac{1}{\sqrt{\alpha(1-\alpha)}} f(k,r,\alpha)^n + o(n),
    \end{align*}
    where the real-valued function $f$ is defined by
    \begin{align*}
        f(k,r,\alpha) \equiv \frac{2}{\alpha^\alpha(1-\alpha)^{1-\alpha}} \left( \frac{1+(1-2\alpha)^k}{4} \right)^r.
    \end{align*}
    For convenience, we simply assume $\alpha^{\alpha}(1-\alpha)^{1-\alpha}=1$ when $\alpha=0$ or $\alpha=1$.
\end{lemma}

\begin{proof}
    Let $Z=Z(\alpha n)$ be the number of pairs of solutions, with distance $\alpha n$ between the two solutions, for a random $k$-XORSAT instance.
    In the other words, given a random instance with linear system representation $Ax=b$,
    \begin{equation*}
        Z = \sum_{\substack{\sigma,\sigma' \in \{0,1\}^n \\ d(\sigma,\sigma')=\alpha n}} \mathbbm{1}(A\sigma=b \text{\;\;and\;\;} A\sigma'=b).
    \end{equation*}
    By linearity of expectation, the expected value of $Z$ is given by
    \begin{align*}
        \expected{Z} &= \sum_{\substack{\sigma,\sigma' \in \{0,1\}^n \\ d(\sigma,\sigma')=\alpha n}} \proba{A\sigma=b \text{\;\;and\;\;} A\sigma'=b}.
    \end{align*}
    Now we consider the calculation of $\proba{A\sigma=b \text{\;and\;}A\sigma'=b}$.
    Since each equation in the linear system are chosen identically and independently, the summand can be written as $(\proba{A_1\sigma=b_1 \text{\;and\;}A_1\sigma'=b_1})^{rn}$, where $A_1x=b_1$ is the first equation in the linear system.
    In addition, by condition probability formula we have
    \begin{align*}
        \proba{A\sigma=b \text{\;and\;}A\sigma'=b}
        &= \left( \proba{A_1\sigma=b_1 \text{\;\;and\;\;} A_1\sigma'=b_1} \right)^{rn} \\
        &= \left( \proba{A_1\sigma=b_1 \text{\;\;and\;\;} A_1\sigma=A_1\sigma'} \right)^{rn} \\
        &= \left( \proba{A_1\sigma=b_1 \;|\; A_1\sigma=A_1\sigma'} \proba{A_1\sigma=A_1\sigma'} \right)^{rn} \\
        &= \left( \frac{1}{2} \proba{A_1\sigma=A_1\sigma'} \right)^{rn}
    \end{align*}

    Since $d(\sigma,\sigma')=\alpha n$, there are $\alpha n$ variables having different values and $(1-\alpha)n$ variables having same values when we compare the assignments $\sigma$ and $\sigma'$.
    The random equation $A_1\sigma=A_1\sigma'$ holds if and only if the equation chooses even number of variables from those $\alpha n$ variables having different values in $\sigma$ and $\sigma'$.
    So, we have
    \begin{align*}
        &\proba{A_1\sigma = A_1\sigma'} \\
        &= \sum_{i=0}^{\lfloor k/2 \rfloor} \proba{2i \text{\;variables with different values in $\sigma$ and $\sigma'$ are chosen by $A_1$}} \\
        &= \sum_{i=0}^{\lfloor k/2 \rfloor} \frac{ \binom{\alpha n}{2i} \cdot \binom{(1-\alpha)n}{k-2i} }{ \binom{n}{k} } \\
        &= \sum_{i=0}^{\lfloor k/2 \rfloor} \binom{k}{2i} \alpha^{2i} (1-\alpha)^{k-2i} \\
        &= \frac{1}{2} (1+(1-2\alpha)^k)
    \end{align*}
    From above formula, we can see that the value of $\proba{A_1\sigma = A_1\sigma'}$ is independent of the choices of $\sigma$ and $\sigma'$, and only depends on the distance between $\sigma$ and $\sigma'$.
    So we have
    \begin{align*}
        \expected{Z}
        &= \sum_{\substack{\sigma,\sigma' \in \{0,1\}^n \\ d(\sigma,\sigma')=\alpha n}} \proba{A_1\sigma=b_1 \text{\;\;and\;\;} A_1\sigma'=b_1}^{rn} \\
        &= \sum_{\substack{\sigma,\sigma' \in \{0,1\}^n \\ d(\sigma,\sigma')=\alpha n}} \left( \frac{1}{2} \proba{A_1\sigma=A_1\sigma'} \right)^{rn} \\
        &= 2^n \binom{n}{\alpha n} \left( \frac{1+(1-2\alpha)^k}{4} \right)^{rn}
    \end{align*}
    By applying the Stirling's approximation for $\binom{n}{\alpha n}$ we have
    \begin{align*}
        \expected{Z}
        &= 2^n \frac{1}{\sqrt{2\pi n}} \frac{1}{\sqrt{\alpha(1-\alpha)}} \left( \frac{1}{\alpha^\alpha(1-\alpha)^{1-\alpha}} \right)^n \left( \frac{1+(1-2\alpha)^k}{4} \right)^{rn} + o(n) \\
        &= \frac{1}{\sqrt{2\pi n}} \frac{1}{\sqrt{\alpha(1-\alpha)}} \left( \frac{2}{\alpha^\alpha(1-\alpha)^{1-\alpha}} \left( \frac{1+(1-2\alpha)^k}{4} \right)^r \right)^n + o(n) \\
        &= \frac{1}{\sqrt{2\pi n}} \frac{1}{\sqrt{\alpha(1-\alpha)}} f(k,r,\alpha)^n + o(n)
    \end{align*}
    where $f(k,r,\alpha)$ is defined by
    \begin{align*}
        f(k,r,\alpha) \equiv \frac{2}{\alpha^\alpha(1-\alpha)^{1-\alpha}} \left( \frac{1+(1-2\alpha)^k}{4} \right)^r.
    \end{align*}
    For convenience, we simply assume $\alpha^{\alpha}(1-\alpha)^{1-\alpha}=1$ when $\alpha=0$ or $\alpha=1$.
\end{proof}

Fix $k \ge 3$ and $r > 0$.
If $f(k,r,\alpha') < 1$ for some $\alpha' \in [0,1]$, then the expectation $\expected{Z(\alpha' n)}$ converges to 0 as $n \rightarrow \infty$.
By Markov’s inequality, we have $\proba{Z(\alpha' n) > 0} \le \expected{\alpha' n}$, and thus $\proba{Z(\alpha' n) > 0}$ also converges to 0 as $n \rightarrow \infty$.
That means the probability of having at least one pair of solutions with distance $\alpha' n$ between them converges to 0.
In the other words, w.h.p. there is no such pair of solutions.
So, if we can find an interval $(u_1,u_2) \subset [0,1]$ such that $f(k,r,\alpha) < 1$ for any $\alpha \in (u_1,u_2)$, we can say that w.h.p. there is no pair of solutions with distance $\alpha n$ between them, for any $\alpha \in (u_1,u_2)$.
In such case, w.h.p. this distance between every pair of solutions is either $\le u_1n$, or $\ge u_2n$, that is, w.h.p. the random instance $\mathbf{\Phi}$ exhibits the overlap gap property with $v_1=u_1n$ and $v_2=u_2n$.

From the formula of the function $f$, we can see that it decreases with $r$, for any fixed $k$ and $\alpha$, illustrated in Figure \ref{fig:moving-f}.
It is more likely to have the OGP if the clause density $r$ is large.
We can further determine the minimal clause density $r_1(k)$ for having the OGP.
This leads to Lemma \ref{lemma:r-1}, with the following proof.

\begin{figure}[ht]
    \centering
    \includegraphics[scale=0.68]{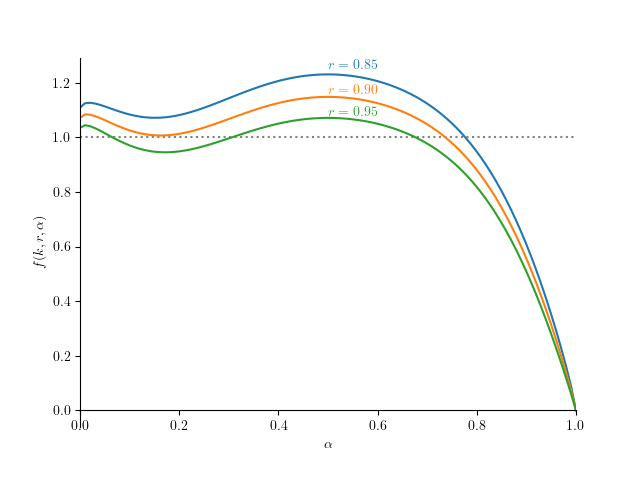}
    \caption{The graph of $f(k,r,\alpha)$ against $\alpha$, with $k=3$ and different values of $r$: (1) $r=0.85$ in blue at the top, (2) $r=0.90$ in brown in the middle, and (3) $r=0.95$ in green at the bottom.}
    \label{fig:moving-f}
\end{figure}

\begin{proof}[Proof of Lemma \ref{lemma:r-1}]
    If $k$ is odd, then we have $f(k,r,1) = 0$.
    In this case, by the continuity of $f$, there exist $u_1 < \frac{1}{2} \le u_2 = 1$ such that $f(k,r,\alpha) < 1$ for any $\alpha \in (u_1,u_2)$.

    Now we assume $k$ is even.
    When $r=0$, it is easy to see that $f(k,r,\alpha) = \frac{2}{\alpha^\alpha(1-\alpha)^{1-\alpha}} > 1$ for any $\alpha \in [0,1]$.
    If we fix $k$ and $\alpha$, then the function $f$ is strictly decreasing with $r$ since $(1+(1-2\alpha)^k)/4 < 1$.
    We aim at finding the minimal positive value of $r$ that guarantees there exists an interval $(u_1,u_2) \subset [0,1]$ such that $f(k,r,\alpha) < 1$ for any $\alpha \in (u_1,u_2)$.
    Furthermore, we know that $f(k,r,\alpha) = f(k,r,1-\alpha)$ for even $k \ge 3$, so it suffices to find the minimal positive value of $r$ that guarantees the existence of such interval $(u_1,u_2)$ within $[0,\frac{1}{2}]$.

    By fixing even $k \ge 3$ and $\alpha \in [0,\frac{1}{2}]$, we can treat $f$ as a strictly decreasing continuous function $f_{k,\alpha}(r)$ of $r$, with $f_{k,\alpha}(0)=2>1$ and $\lim_{r \to \infty}f_{k,\alpha}(r)=0<1$.
    Therefore, for any even $k \ge 3$ and $\alpha \in [0,\frac{1}{2}]$, there exists a unique $r^*(k,\alpha) > 0$ given by
    \begin{align*}
        r^*(k,\alpha) = \frac{1+H(\alpha)}{2-\log_2(1+(1-2\alpha)^k)},
    \end{align*}
    such that
    \begin{align*}
        f_{k,\alpha}(r) &> 1 \text{\quad for } r < r^*(k,\alpha), \\
        f_{k,\alpha}(r) &= 1 \text{\quad for } r = r^*(k,\alpha), \text{ and} \\
        f_{k,\alpha}(r) &< 1 \text{\quad for } r > r^*(k,\alpha)
    \end{align*}
    where $H$ is the binary entropy function $H(x)=-x\log_2(x)-(1-x)\log_2 (1-x)$.
    Then, we can define $r_1(k)$ and $\alpha_1(k)$ by
    \begin{align*}
        r_1(k) = \min_{0\le\alpha\le \frac{1}{2}}r^*(k,\alpha)
        \text{\quad and \quad}
        \alpha_1(k) = \argmin_{0\le\alpha\le \frac{1}{2}}r^*(k,\alpha)
    \end{align*}
    with $r^*(k,\alpha_1(k)) = r_1(k)$.
    Suppose $r > r_1(k)$.
    Since $f$ is strictly decreasing with $r$, we have $f(k,r,\alpha_1(k)) < f(k,r_1(k),\alpha_1(k)) = f_{k,\alpha_1(k)}(r_1(k)) = f_{k,\alpha_1(k)}(r^*(k,\alpha_1(k))) = 1$.
    By the continuity of $f$, there exist $0 \le u_1 < \alpha_1(k)$ and $\alpha_1(k) < u_2 \le \frac{1}{2}$ such that we have $f(k,r,\alpha) < 1$ for any $\alpha \in (u_1,u_2)$.

    Therefore, with Lemma \ref{lemma:first-moment-method}, we know that the expected number $\expected{Z(\alpha n)}$ of pairs of solutions with distance $\alpha n$ between them converges to zero for any $\alpha\in(u_1,u_2)$.
    By the first moment method, w.h.p. there is no pair of solutions with distance $\alpha n$ between them, for any $\alpha\in(u_1,u_2)$.
    In the other words, w.h.p. for any pair of solutions $\sigma$ and $\sigma'$ of $\mathbf{\Phi}$ the distance $d(\sigma,\sigma')$ between them is either smaller than $u_1 n$ or larger than $u_2 n$ for some $u_1$ and $u_2$ with $0 \le u_1 < u_2$.
\end{proof}

\end{appendices}


\end{document}